\documentclass[twoside]{article}
\usepackage{preprint}
\usepackage{amsmath}

\usepackage{amsfonts}
\usepackage{graphicx}
\usepackage{float}
\usepackage{hyperref}

\newtheorem{Theorem}{Theorem}
\newtheorem{Lemma}{Lemma}
\newtheorem{Corollary}{Corollary}
\newtheorem{Proposition}{Proposition}
\newtheorem{Definition}{Definition}
\newtheorem{conjecture}{Conjecture}
\newcommand{\nexists}{\not\exists}
\newcounter{algorithm}
\newenvironment{algorithm}[1][]{\begin{center}\begin{minipage}{0.9\linewidth}\small}{\end{minipage}\end{center}}
\newcommand{\algorithmcaption}[1]{\refstepcounter{algorithm}\noindent\textbf{Algorithm \thealgorithm.} #1\par\vspace{4pt}}
\newenvironment{algorithmic}[1][]{\begin{list}{}{\setlength{\leftmargin}{1.5em}\setlength{\itemsep}{3pt}}}{\end{list}}
\newcommand{\Require}{\item[]\textbf{Require:} }
\newcommand{\Ensure}{\item[]\textbf{Ensure:} }
\newcommand{\State}{\item }
\newcommand{\While}[1]{\item \textbf{while} #1 \textbf{do}\begin{list}{}{\setlength{\leftmargin}{1.5em}}}
\newcommand{\EndWhile}{\end{list}\item[]\textbf{end while}}

\begin{document}
\setlength{\textheight}{8.0truein}
\runninghead{Geometric Algebra: Quantum Gate Decomposition}{Y. Amraoui and Z. Toffano}
\normalsize\textlineskip
\thispagestyle{empty}
\setcounter{page}{1}
\vspace*{0.45truein}
\alphfootnote
\fpage{1}
\centerline{\bf GEOMETRIC ALGEBRA QUANTUM GATE DECOMPOSITION}
\vspace*{0.37truein}
\centerline{\footnotesize YOUSSEF AMRAOUI}
\vspace*{0.015truein}
\centerline{\footnotesize\it CentraleSupélec, Université Paris-Saclay, 91190 Gif-sur-Yvette, France}
\centerline{\footnotesize\it youssef.amraoui@student-cs.fr}
\centerline{\footnotesize ZENO TOFFANO}
\centerline{\footnotesize\it Laboratoire Signaux et Systèmes (L2S), UMR 8506, CentraleSupélec, Université Paris-Saclay, CNRS, 91190 Gif-sur-Yvette, France}
\centerline{\footnotesize\it zeno.toffano@centralesupelec.fr}
\vspace*{0.30truein}
\abstracts{
Quantum gates are traditionally described using matrix and tensor-product formalisms, representations that provide limited geometric intuition.
In this work, we develop a formulation of the Pauli and Clifford groups within the complex Geometric Algebra (GA) framework in order to obtain useful quantum gate decompositions.
We show that the Pauli group is naturally identified with the group of blades up to a global phase, providing an intuitive geometric interpretation of Pauli operators and their commutation relations in terms of oriented subspaces.
We further prove that Clifford operators are generated by products of \(\pi/4\)-Pauli rotors and introduce a greedy Pauli rotor decomposition algorithm whose empirical performance reveals remarkably compact decompositions of Clifford operators.
Finally, we show that Clifford+\(T\) universality also acquires a natural geometric interpretation through \(\pi/8\)-rotors within this framework.
This work highlights Geometric Algebra as a potential geometric tool for quantum computation applications.
}{}{}
\vspace*{10pt}
\vspace*{1pt}\textlineskip
\section{Introduction}

In quantum computing \cite{N&C2010}, quantum gates are the fundamental building blocks of quantum algorithms, mathematically represented by unitary operators that are applied on the states of a quantum processor. The practical implementation of an arbitrary quantum gate is generally challenging, and it is therefore natural to conceive of a quantum computer limited to a finite set of operations that can be physically performed on its qubits, hence the need to study the decomposition of an arbitrary quantum gate into other feasible gates. In order to achieve this decomposition for quantum circuits, the literature yields several different decomposition bases, such as the {Clifford+T} 
 basis \cite{N&C2010} or the Pauli Rotor decomposition \cite{Hegde2016}.

The different decomposition methods rely on a small number of algebraic structures, among which the Pauli and Clifford groups play a particularly central role. Pauli operators provide a discrete basis for describing quantum observables, while Clifford operations describe the class of unitary transformations that preserve this structure under conjugation. Together, these groups form the algebraic backbone of many quantum circuit synthesis techniques and play a fundamental role in the study of quantum computation.

Despite their importance, Pauli and Clifford operators are most often introduced through matrix representations and tensor products, which tend to hide their underlying geometric content. As the number of qubits increases, this representation becomes increasingly combinatorial, offering limited intuition for the structure of multi-qubit operators.

Geometric Algebra (GA), also known as Clifford Algebra, provides a natural mathematical framework for encoding both algebraic and geometric information within a single formalism.  In this framework, vectors, multivectors, and their products admit direct geometric interpretations, while unitary transformations arise naturally as rotors acting \mbox{by conjugation.}

Originally developed to unify vector calculus and geometry, GA has found applications in classical and quantum physics, robotics and computer {science} 
 \cite{Hestenes1966,Doran&Lasenby2003,Breuils&al2022}. 
 There exist various proposals that use GA in quantum computing \cite{Vlasov2001, Havel&Doran2004,Cafaro&al2010,Veyrac&Toffano2024}. Recently, a formalism built on complex Geometric algebras was proposed using the Witt Basis formalism \cite{Hrdina et al. (2022)}.

The purpose of the research presented here is to investigate how the Pauli and Clifford groups can be formulated intrinsically within the complex GA model proposed in \cite{Hrdina et al. (2022)} and to clarify the geometric structure that underlies their standard matrix-based gate descriptions used in quantum computing. Our goal is structural: we aim to show that Pauli operators correspond naturally to blades of a Geometric Algebra and that Clifford operations emerge as symmetry transformations preserving this discrete geometric set.

The paper is organized as follows. In Section \ref{ga}, we introduce the basic notions of Geometric Algebra required for this work. Section \ref{complexification} develops the representation of quantum states and linear operators within this framework. In Section \ref{pauli}, we characterize the Pauli group in Geometric Algebra and discuss its interpretation. Section \ref{clifford} is devoted to the geometric formulation of the Clifford group and the proof of its generation by Pauli rotors; it also discusses the non-Clifford $T$-gate useful for universal quantum computation. In Section \ref{DecAlgo} we present the decomposition algorithm and show the results. We conclude with a discussion of the implications and possible extensions of this approach.

\section{Geometric Algebra: A Brief Introduction} \label{ga}

GA comprises the algebra of quaternions introduced by W.H. Hamilton in 1843  \cite{HamiltonBook} for the description of geometric operations such as rotations and the Grassmann algebra based on the exterior product, which enables the construction of higher-dimensional objects. The synthesis was undertaken by W.K. Clifford around 1870 \cite{Clifford1873,Clifford1878} and generalized into 
a single mathematical structure that he originally called \textit{{Geometric} 
 Algebra} with the introduction of the \textit{geometric product}.

Afterward, for historical reasons, GA was not universally adopted, and the Gibbs vector analysis, which separated scalar and vector products, almost completely replaced the quaternion algebra and the geometric product by the end of the nineteenth century. 
However, in the second half of the twentieth century, some physicists, led by David Hestenes and his followers, used GA methods in classical mechanics, electromagnetism, quantum physics, and special and general relativity \cite{Hestenes1966,Doran&Lasenby2003}.

We begin with a brief yet sufficiently detailed introduction to Geometric Algebra.

This section introduces the basic definitions and tools required to understand the subsequent developments of this paper.
\subsection{Definitions}
Let $p$ be a non-negative integer, and consider the real vector space $(\mathbb{R}^p, \cdot)$ equipped with the usual inner product.

\begin{Definition}[Geometric Algebra]
Let $(\mathcal{A}, +, \cdot)$ be a real unital associative algebra. 
We say that $\mathcal{A}$ is a Geometric Algebra over $(\mathbb{R}^p, \cdot)$ if there exists an injective linear map
\[
f : \mathbb{R}^p \rightarrow \mathcal{A}
\]
such that
\begin{equation}
    \forall x \in \mathbb{R}^p, \qquad f(x)^2 = \|x\|^2 1_{\mathcal{A}}.
\end{equation}
and such that $\mathcal{A}$ is generated by $f(\mathbb{R}^p)$.\\

\vspace{-12pt}
The product in a Geometric Algebra is called the geometric product.
\end{Definition}

It is known that all Geometric Algebras associated with $(\mathbb{R}^p, \cdot)$ are isomorphic. 
Consequently, it is sufficient to work within a single representative algebra.
For this reason, we denote by $\mathcal{G}_p({\mathbb{R}})$ \emph{the} Geometric Algebra generated by $\mathbb{R}^p$, characterized by the \mbox{following {properties:}}
\begin{align*}
    &\mathbb{R} \subset \mathcal{G}_p(\mathbb{R})\quad(\text{so }1_\mathcal{G} =1),\\
    &\mathbb{R}^p \subset \mathcal{G}_p(\mathbb{R})\quad(\text{so }f=\text{Id}).\\
\end{align*}

\vspace{-20pt}
The \textit{geometric product} of two vectors $x, y \in \mathbb{R}^p$ can be decomposed as
\begin{equation}
    xy = x \cdot y + x \wedge y,
\end{equation}
where the first term is the \emph{inner product}. Being commutative, it coincides with the standard scalar product and yields a scalar quantity.

The second term is the \emph{outer product}, defined by
\begin{equation}
x \wedge y = \frac{xy - yx}{2}.
\end{equation}

{The} 
 outer product is anticommutative; it coincides with the Grassmann product and admits a clear geometric interpretation.
It represents an oriented surface element, called a \emph{bivector}. Bivectors are the two-dimensional analog of vectors: they encode not only orientation but also area.

Two vectors commute if (and only if) they are parallel, and they anticommute if (and only if) they are perpendicular to each {other.}

The outer product of $r$ vectors is, when nonzero, called a \textit{$r$-blade}, and elements of the Geometric Algebra are called \textit{multivectors}. A multivector is a linear combination of blades.
\subsection{Basis of Geometric Algebra}
Let $\mathcal{B}=(e_1,...,e_p)$ be the usual orthonormal basis of $\mathbb{R}^p$. It is possible to construct a basis of $\mathcal{G}_p(\mathbb{R})$, using the canonical blades 
\begin{equation}
    e_J = e_{j_1} \cdots e_{j_k} \quad \text{where} \quad J=\{j_1 < \cdots < j_k\} \subset \{1,\dots,p\} \\
\end{equation}
with $e_\emptyset = 1$. So a multivector $A \in \mathcal{G}_p(\mathbb{R})$ is uniquely expressed as 
\begin{equation}\label{expression_mv}
    A = \sum_{J \subset \{1,\dots,p\}} a_Je_J \quad; a_J \in \mathbb{R}
\end{equation}
{We} hence write 
\begin{equation}
    \mathcal{G}_p(\mathbb{R}) = \bigoplus_{r=0}^p \mathcal{G}_p^{(r)}(\mathbb{R}) = \mathcal{G}_p^{+}(\mathbb{R}) \oplus\mathcal{G}_p^{-}(\mathbb{R})
\end{equation}
where $\mathcal{G}_p^{(r)}$ is the vector space spanned by $r$-blades, elements of this subspace are called $r$-\textit{vectors}, and the \text{grade} of such an element is $r$, and $\mathcal{G}_p^{+}(\mathbb{R})$ ($\mathcal{G}_p^{-}(\mathbb{R})$) is the vector space spanned by blades of even (odd) grade. We have
\begin{equation}\label{odd_even_mv}
\mathcal{G}_p^{\varepsilon}(\mathbb{R})\mathcal{G}_p^{\nu}(\mathbb{R})=\mathcal{G}_p^{\varepsilon\nu}(\mathbb{R}) \quad \text{ for } \varepsilon,\nu \in \{+,-\}
\end{equation}
{And}
\begin{align*}
    \dim(\mathcal{G}_p(\mathbb{R})) = 2^p \quad \dim{(\mathcal{G}_p^{(r)}(\mathbb{R}))}= \binom{p}{r}
 \quad \dim(\mathcal{G}_p^{+}(\mathbb{R})) = \dim (\mathcal{G}_p^{-}(\mathbb{R})) =2^{p-1}
\end{align*}

\subsection{Transformations of Multivectors}
We define some operations on multivectors that will be used later.
\begin{itemize}
    \item \textbf{{Projection}}: If $A$ is a multivector, $(A)_r \in \mathcal{G}_p^{(r)}(\mathbb{R})$ is its projection onto the sub-space of $r$-vectors.
    \item \textbf{{Inversion}}: Defined as the automorphism of the Geometric Algebra such that for $A_r \in \mathcal{G}_p^{(r)}(\mathbb{R})$.
    $$A_r^*=(-1)^rA_r$$
    It is the same as replacing vectors with their opposites.
    \item \textbf{{Reversion}}: It is the antiautomorphism of the Geometric Algebra that verifies for a $r$-vector $A_r$
    $$A_r^\dagger = (-1)^\frac{r(r-1)}{2}A_r$$
    It is the same as reversing the order of vectors in the product of a $r$-vector.
\end{itemize}
\section{Complexification of the Geometric Algebra} \label{complexification}
The representation of quantum states and elementary quantum gates adopted in this section follows the complex Clifford algebra formalism introduced by Hrdina et~al.~\cite{Hrdina et al. (2022)}. In particular, we use the Witt basis construction and the associated representation of computational basis states.

Building upon this framework, we establish in Theorem \ref{isomorphism} an explicit correspondence between arbitrary linear operators acting on the quantum register and left multiplication by multivectors. This result provides a complete algebraic identification between matrix operators and elements of the Geometric Algebra, which will serve as the foundation for the developments of the subsequent sections.

\subsection{Complexification, the Witt Basis}
The complexification of the Geometric Algebra \cite{Hrdina et al. (2022),Veyrac&Toffano2024} 
 means allowing for complex coefficients. The same generators of the \textit{real} Geometric Algebra are henceforth generators of the \textit{complex} GA. An element of the complex Geometric Algebra can be written \mbox{naturally as} 
$$A=X+iY$$
where $X$ and $Y$ are real multivectors, and $i$ is the imaginary number. We thus write 
$$\mathcal{G}_p(\mathbb{C})=\mathcal{G}_p(\mathbb{R}) \oplus i \mathcal{G}_p(\mathbb{R})=\text{Span}_\mathbb{C}\{e_J | J \subset \{1,\cdots ,p\}\}$$

We extend the operations on real multivectors onto complex multivectors by linearity while taking into account that the grade of a complex number is always $0$, its inversion is itself, and its reversion is its usual complex conjugate. We can also define an inner product of multivectors, for $A, B \in \mathcal{G}_p(\mathbb{C})$
\begin{align}
    \langle A|B\rangle = \big( A^\dagger B \big)_0
\end{align}
{In fact,} this inner product is a Hermitian product, if $A$ and $B$ are written as in \eqref{expression_mv} (with complex coefficients)
$$\langle A|B \rangle = \sum_{J \subset \{1,\cdots,p\}} a_J^\dagger b_J$$

Let us assume this time that the dimension of the generating vector space is even, so $p=2n$. Instead of considering the generating vector basis $\big\{e_j, ie_j | j \in \{1,\cdots 2n\} \big\}$, we introduce the so-called \textit{Witt basis}, that is the basis $(f_j,f_j^\dagger)$ such that 
\begin{align}
    f_j = \frac{1}{2}\big(e_j - i e_{n+j}\big) ,\text{ }j=1,\cdots,n \\
    f_j^\dagger=\frac{1}{2}\big(e_j + i e_{n+j}\big),\text{ }j=1,\cdots,n
\end{align}

Note that the vectors of the Witt basis verify the Grassmann identities 
\begin{equation}\label{grassmann_id}
    f_jf_k+f_kf_j=0 ,\text{ }j,k=1,\cdots,n
\end{equation}
{As well} as the duality identities 
\begin{equation}\label{duality_id}
    f_jf_k^\dagger + f_k^\dagger f_j = \delta_{jk}, \text{ }j,k=1,\cdots,n
\end{equation}

From now on, we denote by $\mathcal{G}_n$ the $2n$-complex Geometric Algebra 
\begin{equation}\mathcal{G}_n = \mathcal{G}_{2n}({\mathbb{C}})\end{equation}
and by $\mathcal{F} = (f_1,..,f_n,f^\dagger_1,...,f_n^\dagger)$ the associated Witt basis.
\subsection{Dirac Formalism in Geometric Algebra}
We consider a system of $n$-qubits, the $n$-qubit states can be represented in the Geometric Algebra $\mathcal{G}_n$ by \cite{Hrdina et al. (2022)}
\begin{align}
    |j\rangle=|j_1\cdots j_n \rangle \leftrightarrow \varphi_j =  (f_1^\dagger)^{j_1} \cdots (f_n^\dagger)^{j_n}I \quad
\end{align}
where $I=f_1f_1^\dagger\cdots f_nf_n^\dagger$.

Furthermore, the unitary transformations on states are expressed by left multiplication by \textit{unitary multivectors} $\lambda \in \mathcal{G}_n$ verifying
\begin{align}
    \lambda\lambda^\dagger=1  
\end{align}

\subsection{Case of 1-Qubit System}
The simple system of one qubit is represented by 
\begin{align}
    \psi=aff^\dagger+bf^\dagger
\end{align}
where $a,b \in \mathbb{C}$ such that $|a|^2+|b|^2=1$.

{The} 
 known 1-qubit gates in quantum computing are 

\vspace{-15pt}
\begin{align}
    \text{Pauli-X} \quad\lambda_X=&f+f^\dagger=e_1 \\
    \text{Pauli-Y}\quad\lambda_Y=&-i(f-f^\dagger) =-e_2\\
    \text{Pauli-Z}\quad\lambda_Z=&ff^\dagger-f^\dagger f=ie_1e_2\\
    \text{Phase}\quad\lambda_S =& ff^\dagger + if^\dagger f= e^{i\frac{\pi}{4}}(\frac{1+e_1e_2}{\sqrt2})\\
    \text{Hadamard}\quad\lambda_H=&\frac{1}{\sqrt2}(f+ff^\dagger +f^\dagger -f^\dagger f) = e_1(\frac{1 +ie_2}{\sqrt2})
\end{align}

In the case of one qubit, the link between 
standard matrices and multivectors is \linebreak  as follows.
\begin{align}
    \begin{pmatrix}
        a & b \\
        c & d \\
    \end{pmatrix}
    \equiv aff^\dagger+bf+cf^\dagger+df^\dagger f
\end{align}

The previous constructions allow us to represent quantum states as multivectors and
elementary quantum gates as specific elements of the Geometric Algebra.
A natural question that follows is whether this correspondence extends beyond
elementary gates to arbitrary linear transformations acting on the quantum register.
In other words, we ask whether the entire matrix algebra of linear operators on the
Hilbert space can be faithfully represented within the Geometric Algebra formalism.

The following theorem extends the previous correspondence to arbitrary linear operators that act on the quantum register, going beyond elementary gates inherited from the framework of Hrdina et al. \cite{Hrdina et al. (2022)}. 

\begin{Theorem}\label{isomorphism}
    Every linear transformation in the space of the quantum register is uniquely represented by left multiplication by a multivector in $\mathcal{G}_n$. For a linear transformation $M =(m_{k,l}) \in \mathcal{M}_{2^n}(\mathbb{C})$, its associated multivector is given by 
    \begin{equation}\label{matrix_mv}
        \lambda_M = \sum_{0\leq l,k \leq 2^n-1} m_{k,l} \varphi_k \varphi_l^\dagger
    \end{equation}
\end{Theorem}
\begin{proof}
    {We} 
 start by computing the multivector $f_k^\dagger\varphi_l$
    \vspace{-6pt}

\begin{align}
        f_k^\dagger\varphi_l =& f_k^\dagger (f_1^\dagger)^{l_1} \cdots (f_n^\dagger)^{l_n}I \\ \nonumber
        =&(-1)^{l_1+\cdots l_{k-1}}(f_1^\dagger)^{l_1}\cdots(f_k^\dagger)^{1+l_k}\cdots (f_n^\dagger)^{l_n}I ~~ \text{(Because }(f_j)_j \text{ anticommute from \eqref{grassmann_id}})\\ \nonumber
        \end{align}

\vspace{-20pt}
        {Since} $(f_k^\dagger)^2=0$ (from \eqref{grassmann_id}) we get
        \begin{equation}\label{fdagger_operation}
        f_k^\dagger\varphi_l=\begin{cases}
            0 &\text{ if }l_k = 1 \text{ (because } (f_k^\dagger)^2=0 \text{ from \eqref{grassmann_id})}\\ 
            (-1)^{l_1+\cdots l_{k-1}}\varphi_{l+2^{k-1}}  &\text{ else} 
        \end{cases}
    \end{equation}
{We} compute next the multivector $f_k\varphi_l$
\vspace{-6pt}

\begin{align}
    &f_k\varphi_l = f_k (f_1^\dagger)^{l_1} \cdots (f_n^\dagger)^{l_n}I \\ \nonumber
    =&(-1)^{l_1+\cdots l_{k-1}}(f_1^\dagger)^{l_1}\cdots f_k(f_k^\dagger)^{l_k}\cdots (f_n^\dagger)^{l_n}I ~~\text{Because }f_k\text{ and }f_j^\dagger,j\neq k \text{ anticommute from \eqref{duality_id}}\\ \nonumber
\end{align}

\vspace{-20pt}
{If} $l_k=0$, then we get
\begin{align}
    f_k\varphi_l = (-1)^{l_1+\cdots l_n}(f_1^\dagger)^{l_1}\cdots (f_n^\dagger)^{l_n}f_kI
\end{align}
{However}
\begin{align}\nonumber
    f_kI =& f_kf_1f_1^\dagger\cdots f_nf_n^\dagger\\ \nonumber
    =&f_1f_1^\dagger \cdots f_k^2f_k^\dagger \cdots f_nf_n^\dagger\\ 
    =& f_1f_1^\dagger \cdots 0f_k^\dagger \cdots f_nf_n^\dagger\\ \nonumber
    =&0 
\end{align}
{So} \begin{align}
    f_k\varphi_l=0
\end{align}
{Otherwise}, we use \eqref{duality_id} to get
\begin{align}\nonumber
    f_k\varphi_l &=(-1)^{l_1+\cdots l_{k-1}}(f_1^\dagger)^{l_1}\cdots (1-f_k^\dagger f_k)\cdots (f_n^\dagger)^{l_n}I\\ 
    &=(-1)^{l_1+\cdots l_{k-1}}(f_1^\dagger)^{l_1}\cdots (f_k^\dagger)^0\cdots (f_n^\dagger)^{l_n}I-(-1)^{l_1+\cdots l_n}(f_1)^{l_1}\cdots (f_n^\dagger)^{l_n}f_kI\\ \nonumber
    &=(-1)^{l_1+\cdots l_{k-1}}\varphi_{l-2^{k-1}}
\end{align}
{So} we conclude that
\begin{equation}
\label{f_operation}
    f_k\varphi_l=\begin{cases}
            0 &\text{ if }l_k = 0\\
            (-1)^{l_1+\cdots l_{k-1}}\varphi_{l-2^{k-1}}  &\text{ else}
        \end{cases}
\end{equation}

When seen in quantum context, $f_k$ and $f_k^\dagger$ can be interpreted as modified versions of creation and annihilation operators.

Next, we shall verify that $\varphi_k^\dagger\varphi_l = \delta_{k,l}I$.

Let $k\neq l \leq 2^n$, let $1\leq j \leq n$ be the smallest integer such that $l_j \neq k_j$ (so \linebreak  $l_1=k_1,...,l_{j-1}=k_{j-1}$), since we can always switch between $k$ and $l$ by considering $(\varphi_k^\dagger \varphi_l)^\dagger =\varphi_l^\dagger \varphi_k$, we can assume without loss for generality that $l_j=0$ and $k_j=1$. Following \eqref{f_operation} \mbox{we get}
\begin{align} 
    \varphi_k^\dagger \varphi_l =& I^\dagger f_n^{k_n}\cdots f_1^{k_1} \varphi_l\\ 
    =&(-1)^aIf_n^{k_n} \cdots f_{j+1}^{k_{j+1}}f_j\varphi_{l'} \nonumber
\end{align}
where $l' = l - k_1 - k_22-\cdots k_{j-1}2^{j-2}$, since $l'_j=l_j=0$, we get $f_j\varphi_{l'}=0$. Therefore 
\begin{align}
    \varphi_k^\dagger \varphi_l = 0
\end{align}
{On} the other hand, we have in the case of $k=l$
\begin{align}
    \varphi_{k}^\dagger\varphi_k =& I\varphi_0I=I^3=I
\end{align}
{We} conclude that
\begin{equation}
    \varphi_k^\dagger\varphi_l = \delta_{k,l}I
\end{equation}
{Since} $\varphi_j I = \varphi_j$, we get multivector equivalent of the elementary matrices 
\begin{align}
    \big(\varphi_k\varphi_l^\dagger\big)\varphi_j = \delta_{j,l}\varphi_k
\end{align}
{which} proves the existence part of the theorem. It also proves the given formula in \eqref{matrix_mv}. As for the uniqueness, we note that the mapping $M \mapsto \lambda_M$ is an injective morphism from $\mathcal{M}_{2^n}(\mathbb{C})$ to $\mathcal{G}_n$. Since both are of dimension $2^{2n}$, the mapping is an isomorphism. 
\end{proof}

This result shows that Geometric Algebra provides a complete and faithful representation
of linear transformations on quantum registers.
In particular, it allows us to transfer familiar algebraic properties of matrices directly
into the Geometric Algebra setting.
We summarize below the main structural consequences of this correspondence, which will
be repeatedly used in the remainder of the {paper.} 
\begin{align}
    & \quad \lambda_{M+N} = \lambda_{M}+\lambda_{N}\\  
    & \quad\lambda_{MN} = \lambda_{M}\lambda_{N}\\ 
    & \quad \lambda_{M^\dagger} = \lambda_M^\dagger \quad\text{where } M^\dagger\text{ is the Hermitian transpose of }M\\ 
    & \quad  \lambda_{Id}=1\\ 
    & \quad\text{tr}(M) = 2^n(\lambda_M)_0 
\end{align}
\subsection{Tensor Product of Quantum Gates}

In the complex Clifford algebra framework \cite{Hrdina et al. (2022)}, parallel quantum gates are represented as tensor products. An \( n \)-qubit state \( |i_{1} \dots i_n\rangle =|i_{1}\rangle \otimes \dots \otimes  |i_{n}\rangle \) is the geometric product of individual qubit representations, so {\( |\phi\rangle \otimes |\psi\rangle \)} 
becomes \( \phi \psi \), with \( \phi, \psi \) in orthogonal subspaces.

For gates \( \lambda \otimes \mu \), acting on {\( |\phi\rangle \otimes |\psi\rangle \)} gives \( \lambda |\phi \rangle\otimes \mu |\psi\rangle \), represented as \( \lambda \phi \mu \psi \). Due to the geometric product's noncommutativity, for blades \( e_A, e_B \), we have:
\[
e_A e_B = (-1)^{|A||B|-|A \cap B|} e_B e_A.
\]

\begin{Theorem}[\cite{Hrdina et al. (2022)}]\label{tensor}
The tensor product \( \lambda_1 \otimes \dots \otimes \lambda_n \), with \( \lambda_k \in \{ f_k f_k^\dagger, f_k^\dagger f_k, f_k, f_k^\dagger \} \), is \( (-1)^s \lambda_1 \dots \lambda_n \), where \( s = \sum_i |S_i| \), and:
\[
S_i = \{ \ell < i : \lambda_\ell = f_\ell \text{ or } \lambda_\ell = f_\ell^\dagger f_\ell \} \quad \text{if } \lambda_i = f_i \text{ or } \lambda_i = f_i^\dagger.
\]
\end{Theorem}

\section{The Pauli Group in Geometric Algebra} \label{pauli}
The Pauli group and Pauli matrices are fundamental in quantum computing \cite{N&C2010}, as they form a basis for describing quantum states and operations. The Pauli matrices ($X$, $Y$, $Z$) represent single-qubit operations, enabling key transformations such as bit flips and phase shifts. More generally, tensor products of Pauli matrices generate the Pauli group, which plays a central role in quantum computing.

Particularly important are commuting subsets of Pauli operators. Maximal commuting Pauli subsets define simultaneously diagonalizable structures and naturally generate generalized rotations acting on multi-qubit systems. Such decompositions have important applications in quantum circuit synthesis and Hamiltonian simulation, where unitary operators may be approximated or implemented through products of commuting Pauli rotations. In particular, Hegde {et al.} introduced a decomposition framework based on Pauli decompositions over commuting subsets (PDCS), demonstrating efficient implementations of quantum gates, state preparation procedures, and quantum simulations using commuting Pauli rotors \cite{Hegde2016}.

For an $n$-qubit system, maximal commuting subsets of the Pauli group contain at most $2^n-1$ nontrivial Pauli operators. These applications illustrate that the algebraic structure of commuting Pauli subsets can be exploited effectively for unitary approximation and quantum gate synthesis.

The Pauli matrices for $1$-qubit are given by
$$\quad X= \begin{pmatrix}
    0 & 1 \\
    1 & 0
\end{pmatrix}; \quad Y= \begin{pmatrix}
    0 & -i \\
    i & 0
\end{pmatrix};\quad Z=\begin{pmatrix}
    1 & 0 \\
    0 & -1
\end{pmatrix}$$

The single-qubit Pauli group is defined as the group generated by the Pauli matrices 
$$\mathcal{P}_1 =\langle\{X,Y,Z\} \rangle=\big\{i^kU|k\in \{0,1,2,3\};U\in \{I,X,Y,Z\}\big\}$$

For a system of $n$-qubits, the Pauli group is defined as the group generated by all possible tensor products of different $1$-qubit Pauli matrices
\begin{align} \nonumber
   \mathcal{P}_n =& \langle\big\{ U_1 \otimes \cdots \otimes U_n | U_i \in \{X,Y,Z\} \big\}\rangle\\
=& \big\{i^k U_1 \otimes \cdots \otimes U_n| k \in \{0,1,2,3\};U_i \in \{I,X,Y,Z\} \big\} 
\end{align}

We recall from the elementary unitaries in GA
\begin{align}
    \lambda_X=e_1 \quad
    \lambda_Y=-e_2 \quad \lambda_Z=ie_1e_2
\end{align}

So the single-qubit Pauli group can be seen as the natural group of blades:
$$\mathcal{P}_1= \{ i^ke_J | k \in\{0,1,2,3\}, J \subset \{1,2\}\}$$

Since the $1$-qubit Pauli operators admit simple multivector representatives, it is natural to ask
whether the $n$-qubit Pauli group can be characterized simply within the GA.
The next Theorem provides a complete characterization of the Pauli group as a group of
blades, up to a global phase.

\begin{Theorem}
    In $\mathcal{G}_n$, the Pauli group is the group of blades, up to a global phase multiple of $\pi/2$ 
    \begin{equation}
        \mathcal{P}_n = \{ i^ke_J| \quad 0\leq k \leq 3, J \subset \{1,\dots,2n\}\}
    \end{equation}
\end{Theorem}
\begin{proof}
    Let us define the following groups
    \begin{align}
        \mathcal{P}_n^M =&\langle\big\{ i^kU_1 \otimes \cdots \otimes U_n | 0\leq k\leq 3 ,U_i \in \{I,X,Y,Z\} \big\}\rangle \subset \mathcal{M}_{2^n}(\mathbb{C})\\
        \mathcal{P}_n^G=&\{ \lambda_P |\quad P \in \mathcal{P}_n^M\}\\
        \mathcal{E}=&\{ i^ke_J| \quad 0\leq k \leq 3, J \subset \{1,\dots,2n\}\}
    \end{align}
    {The} goal is to prove that $\mathcal{P}_n^G=\mathcal{E}$.
    
    {From} 
 the definition of the Witt basis, we know that $e_j = f_j +f_j^\dagger $ and $e_{j+n} = i(f_j - f_j^\dagger)$ for $j=1,...,n$. From \eqref{fdagger_operation} and \eqref{f_operation} we have
    \begin{align}
        e_j \varphi_l&= (-1)^{l_1+ \cdots l_{j-1}} \begin{cases}
            \varphi_{l - 2^{j-1}} & \text{ if } \quad l_j=1\\
            \varphi_{l + 2^{j-1}} & \text{ if } \quad l_j=0\\
        \end{cases}\\ 
        e_{j+n}\varphi_l &=i(-1)^{l_1+ \cdots l_{j}} \begin{cases}
            \varphi_{l - 2^{j-1}} & \text{ if } \quad l_j=1\\
            \varphi_{l + 2^{j-1}} & \text{ if } \quad l_j=0\\
            \end{cases}
    \end{align}
    
{Notice} that the mapping
$$\varphi_l \mapsto\begin{cases}
    \varphi_{l-2^{j-1}} &\text{ if} \quad l_j=1 \\
     \varphi_{l+2^{j-1}}&\text{ if} \quad l_j=0
\end{cases}$$
corresponds to the $j$-th qubit flip, in other words to the operation $X_j$. The $(-1)^{l_1+\cdots}$ factor corresponds to a diagonal matrix whose entries are either $1$ or $-1$, more precisely, we can conclude that the multivector $e_j , j=1,\dots,2n$ corresponds to the matrix 
\begin{align}
    A_j = i^{\varepsilon^j}S_{\alpha^j_1,\dots \alpha^j_{2^n}}X_{j -\varepsilon^jn}
\end{align}
where 
\begin{align}
    \varepsilon^j &= 1_{j> n}\\ 
    \alpha_l^j &=l_1+\cdots l_{j-1}+\varepsilon^jl_j\\
    S_{\alpha_1^j,\dots\alpha_{2^n}^j} &= \text{diag}((-1)^{\alpha_1^j},\cdots,(-1)^{\alpha_{2^n}^j})\\ \nonumber
\end{align}

\vspace{-20pt}
{Let} us verify that $S_{\alpha_1,\cdots \alpha_{2^n}}$ is a Pauli matrix; for that, we consider $\mathcal{S}$ as the set of possible matrices of that form, i.e., the diagonal matrices of diagonal coefficients either $-1$ or $1$. On the other hand, we consider $\mathcal{Z}$ the set of Pauli matrices formed only from the tensor products of $Z$ and $I_2$. Since $Z$ is diagonal with diagonal coefficients being $1$ and $-1$, then $\mathcal{Z} \subset \mathcal{S}$, however both sets are of cardinality $2^n$, thus we conclude that $\mathcal{S} = \mathcal{Z} \subset \mathcal{P}_n^M$.
Thus $A_j \in \mathcal{P}_n^M$, which means that
$$\forall j=1,\dots,2n: e_j \in \mathcal{P}_n^{G}$$
{Therefore}
$$\mathcal{E} \subset\mathcal{P}_n^{G}$$
{Since} both groups have cardinality $4^{n+1}$, we conclude that
$$\mathcal{E} =\mathcal{P}_n^{G}$$
\end{proof}

The characterization of the Pauli group as a group of blades highlights the importance of understanding commutation relations between multivectors.
In the Pauli formalism, commutation and anticommutation relations encode fundamental information about compatibility of observables and error propagation.

\begin{Proposition}\label{prop_commutation}
    Let $J,K \subset \{1,\dots,2n\}$,
    \begin{align}
        e_Je_K = e_Ke_J \iff |K\cap J| = |K||J| \mod 2
    \end{align}
\end{Proposition}
\begin{proof}
    The result follows by counting the sign changes produced when exchanging the vectors composing the two blades. The complete derivation is provided in \mbox{Appendix \ref{prop_commutation_proof}.}
\end{proof}
\subsection*{Interpretation}

The characterization of the Pauli group as a group of blades gives a natural geometric interpretation to Pauli operators. In this framework, Pauli elements are no longer viewed merely as matrices acting on a Hilbert space, but rather as oriented geometric objects inside the Geometric Algebra.

Vectors represent oriented directions, bivectors represent oriented planes, and higher-grade blades represent oriented higher-dimensional geometric objects. The Pauli group therefore appears as a discrete family of geometric primitives together with their \mbox{orientation reversals.}

For example, the subgroup generated by the vector \(e_1\) and the bivector \(e_{23}\) contains the elements
\[
\{\pm 1,\pm e_1,\pm e_{23},\pm e_{123}\}
\]
{Geometrically}, this subgroup contains the direction \(e_1\), the oriented plane \(e_{23}\), and the oriented volume element \(e_{123}\) generated by their geometric product. The negative elements correspond naturally to orientation reversals of the same geometric objects.

This point of view also gives a geometric interpretation of commutation relations, illustrated in Figure~\ref{fig:geometric_commutation}. In Figure~\ref{fig:geometric_commutation}a, the vector $e_2$ is orthogonal to the plane generated by the bivector $e_{13}$. Rotations in this plane therefore leave $e_2$ invariant, which corresponds algebraically to the commutation relation 
\[ e_2e_{13}=e_{13}e_2 \] 
{By} contrast, Figure~\ref{fig:geometric_commutation}b shows a vector contained in the rotation plane. In this case, the vector is transformed nontrivially by the rotation, corresponding to an anticommutation relation. For instance, 
\[ e_3e_{13}=-e_{13}e_3. \] 
{Thus}, commutation and anticommutation of Pauli blades can be interpreted as geometric compatibility or incompatibility between the corresponding oriented subspaces.

\vspace{-6pt}
\begin{figure}[htbp]
\hspace{-20pt}
\begin{minipage}{0.43\textwidth}\label{a}
    \centering
    \includegraphics[width=0.7\linewidth]{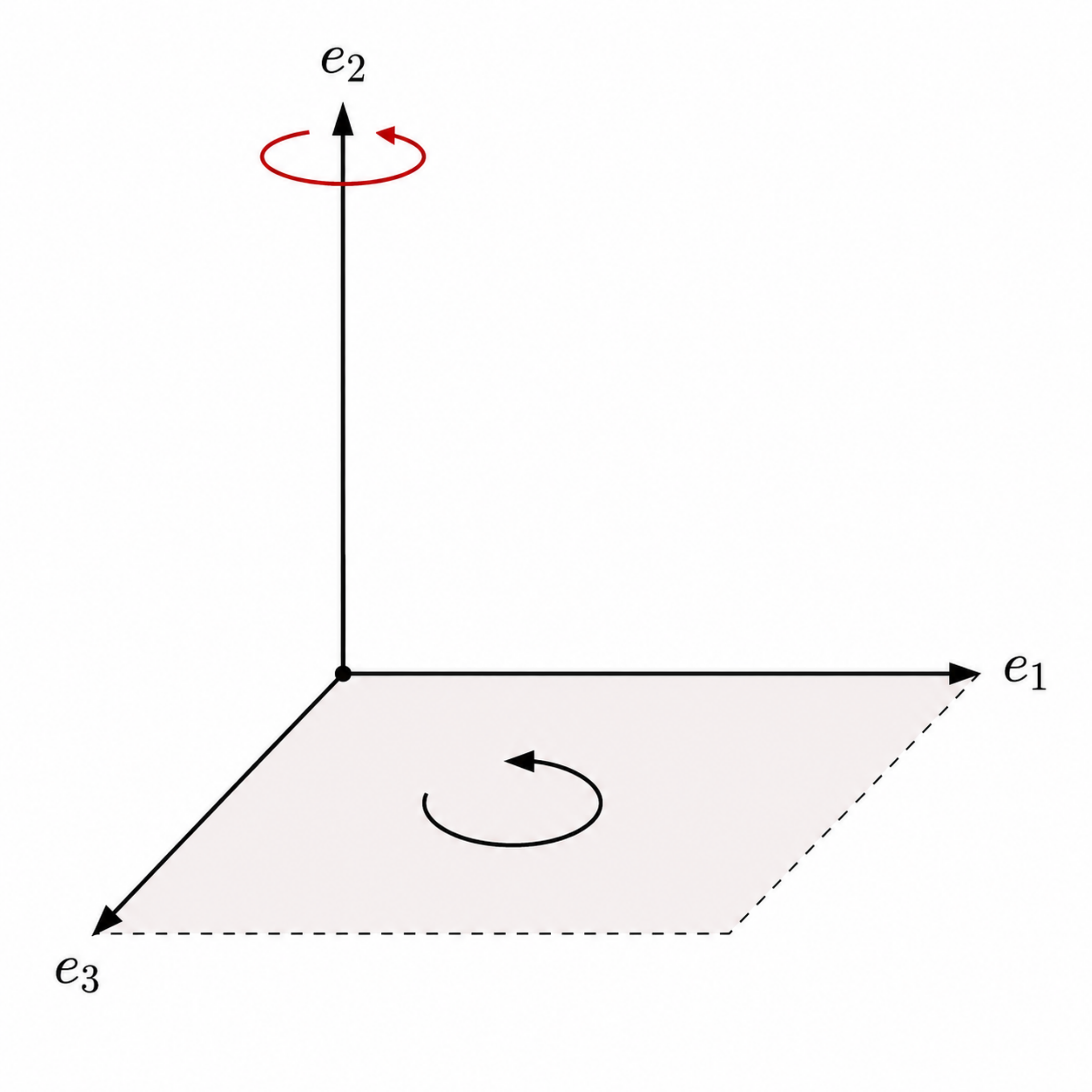}
    \textbf{(a)}
\end{minipage}
\hfill
\begin{minipage}{0.43\textwidth}\label{b}
    \centering
    \includegraphics[width=0.7\linewidth]{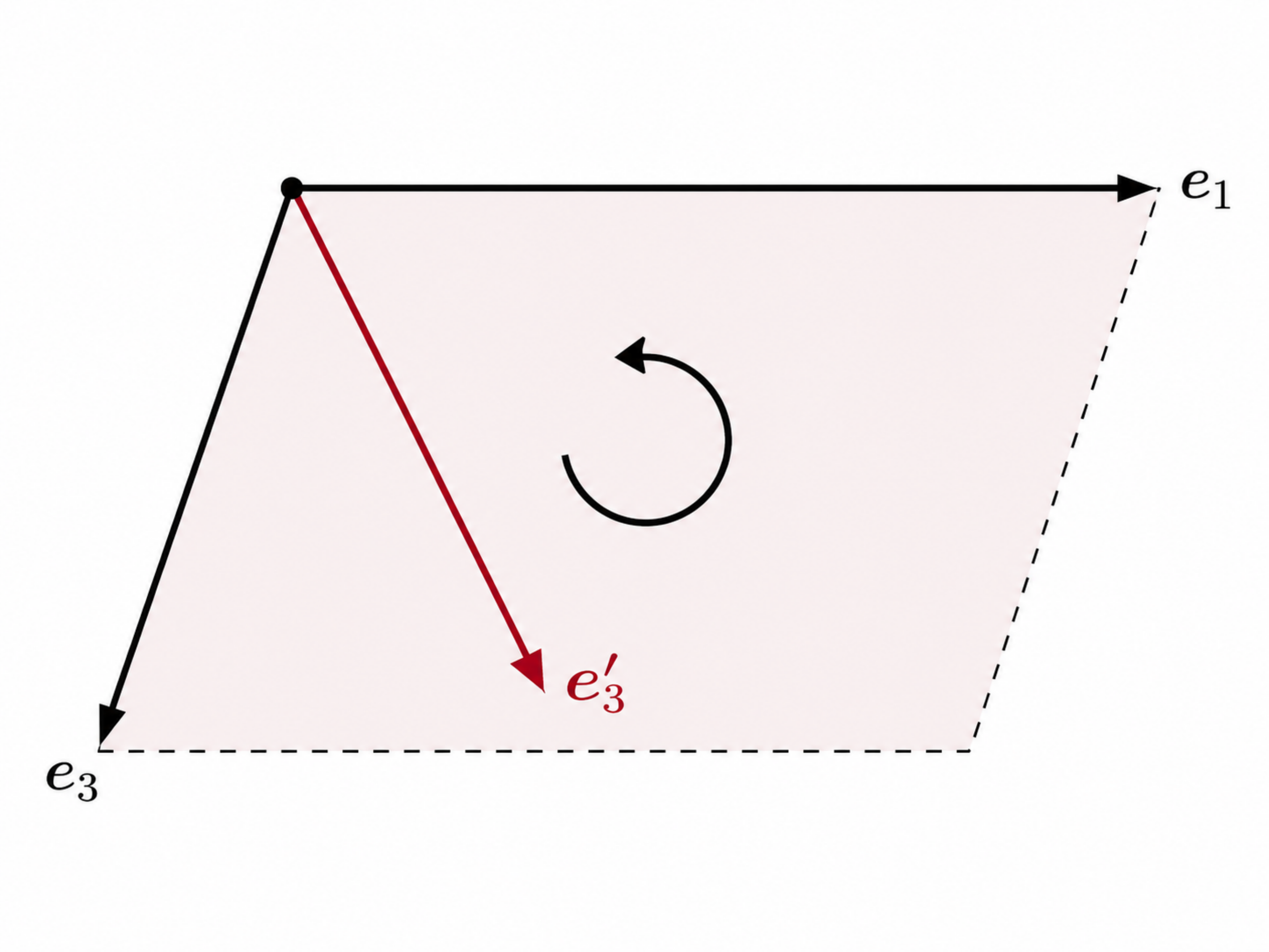}
    \textbf{(b)}
\end{minipage}

\fcaption{
Geometric interpretation of commutation and anticommutation of Pauli blades.
(\textbf{a})~A vector orthogonal to the rotation plane is left invariant, corresponding
to a commuting relation. 
(\textbf{b}) A vector contained in the rotation plane is transformed nontrivially,
corresponding to an \mbox{anticommuting relation.}
}
\label{fig:geometric_commutation}
\end{figure}

More generally, Proposition \ref{prop_commutation} shows that commutation between Pauli blades is determined entirely by the parity of their geometric overlap. In this sense, compatibility relations between Pauli operators become geometric incidence relations between \mbox{oriented subspaces.}

From this perspective, the Pauli group can be viewed as a discrete geometric structure embedded inside the GA, while Clifford operators, the subject of the next Section \ref{clifford}, act as symmetry transformations preserving this structure under conjugation.

This geometric viewpoint constitutes one of the main motivations for employing Geometric Algebra. While the same operators can be represented algebraically in other formalisms, the blade structure provides an immediate interpretation of Pauli operators and their compatibility relations in terms of oriented subspaces and their intersections.

\section{The Clifford Group in Geometric Algebra} \label{clifford}
The Clifford group plays a central role in quantum information theory \cite{N&C2010}, as it describes the class of unitary transformations that preserve the Pauli group under conjugation.

These transformations arise naturally in quantum circuit design and in the stabilizer formalism. Since Clifford operators preserve the Pauli group under conjugation, they occupy a central role in quantum information theory and quantum gate synthesis.

However, the Clifford group alone is not computationally universal \cite{N&C2010}. Quantum circuits generated solely by Clifford operations can be efficiently simulated classically through the Gottesman--Knill Theorem \cite{Gottesman (1998),Aaronson and Gottesman (2004)}. In order to achieve universal quantum computation, one must supplement the Clifford group with at least one non-Clifford gate, most commonly the \(T\)-gate \cite{N&C2010}. The resulting Clifford+\(T\) gate set forms a universal generating set for quantum computation and plays a fundamental role in modern fault-tolerant architectures.

A major challenge in quantum circuit synthesis is therefore the approximation of arbitrary unitary operators using Clifford+\(T\) circuits while minimizing both circuit depth and \(T\)-count \cite{Kliuchnikov et al. (2013),Ross and Selinger (2016),Amy et al. (2020)}, since \(T\)-gates typically constitute the dominant resource cost in fault-tolerant implementations. This problem has motivated extensive research on efficient decomposition algorithms, gate optimization techniques, and compact representations of Clifford operations.

Consequently, understanding the internal structure and decomposition properties of Clifford operators remains an important problem both mathematically and algorithmically. In this section, we investigate rotor decompositions of Clifford operators within the GA framework developed previously.

The Clifford group is usually defined as 
\begin{align}
    \mathcal{C}_n = \{ Q \in \mathcal{U}_{2^n}(\mathbb{C}) |\quad  Q\mathcal{P}_nQ^\dagger = \mathcal{P}_n \}
\end{align}

It is known that such a group is generated by the 1-qubit gates $S$, the Hadamard $H$, and also the 2-qubit CNOT gate. The Clifford group can be viewed as the group of symmetries of the Pauli group under conjugation.
In the Geometric Algebra setting, this amounts to identifying the unitary multivectors whose adjoint action preserves the set of Pauli blades.
Our main goal in this section is to characterize this group intrinsically and to describe its generators.
The following Theorem states the main structural result, whose proof occupies the
remainder of this section.
\begin{Theorem}\label{clifford_grp}
    Let $\mathcal{C}_n$ be the Clifford group in $\mathcal{G}_n$, defined as 
    $$\mathcal{C}_n=\{ \lambda \in \mathcal{G}_n | \quad \lambda \lambda^\dagger =1 \text{ and }\lambda\mathcal{P}_n\lambda^\dagger = \mathcal{P}_n\}$$
    then $\mathcal{C}_n$ is generated, up to a global phase, by the rotors of the form 
    \begin{equation}
        \rho_b = \exp\left( \frac{\pi}{4}b\right) \quad b \in\mathcal{P}_n \text{ and }b^2 =-1
    \end{equation}
\end{Theorem}
Let us start by defining some important groups, mainly $\mathcal{U}_n$, the group of unitary multivectors, and $\mathcal{R}_n$, the group generated by the rotors $\rho_b$. We also define the sets $\mathcal{P}_n^+$ (and $\mathcal{P}_n^-$) as the set of Pauli multivectors of square $1$ (and $-1$).

To establish Theorem \ref{clifford_grp}, we first require three auxiliary results describing the action of Pauli rotors on Pauli blades and their products. For the sake of readability, the detailed proofs of these intermediate results have been deferred to Appendix \ref{app1}.

\begin{Lemma}\label{rotor_b_effect}
    Let $b \in \mathcal{P}_n^-$ and $\pi \in \mathcal{P}_n$, then 
    \begin{equation}
        \rho_b\pi\rho_b^\dagger = \begin{cases}
            \pi & \text{ if } \pi b = b\pi\\
            b\pi & \text{ if } \pi b =- b\pi
        \end{cases}
    \end{equation}
    {Consequently}, 
\begin{equation}
    \mathcal{R}_n \subset \mathcal{C}_n
\end{equation}
\end{Lemma}

\begin{proof}
Since $\rho_b=(1+b)/\sqrt2$, the conjugation formula reduces to two cases according to whether $b$ and $\pi$ commute or anticommute. In both situations, the image remains a Pauli blade. The detailed calculation is provided in Appendix \ref{rotor_b_effect_proof}.
\end{proof}

Lemma \ref{rotor_b_effect} shows that Pauli rotors act on Pauli blades by either fixing them or mapping
them to another Pauli blade.
In particular, this proves that the group generated by these rotors is contained in the Clifford group.
We now investigate in the converse direction, namely whether rotors can be used to reduce arbitrary Pauli elements to \mbox{canonical representatives}.

\begin{Lemma}\label{rotor_e1}
    Let $b \in\mathcal{P}_n^+$ such that $b \neq \pm 1$, then
    \begin{equation}
        \exists u \in \mathcal{R}_n : 
            ubu^\dagger=e_1
    \end{equation}
\end{Lemma}
\begin{proof}
The proof relies on constructing Pauli rotors adapted to the support of the blade \(b=i^\varepsilon e_J\). In the even-grade case, these rotors reduce \(b\) to a basis vector and then map it to \(e_1\). In the odd-grade case, an auxiliary vector outside the support of \(b\) is first introduced to reduce the problem to the even case. The complete construction is given in \mbox{Appendix \ref{rotor_e1_proof}.}
\end{proof}    

The previous lemma establishes that any Pauli multivector of square $1$ can be conjugated to a fixed basis vector.
This observation naturally extends to pairs of anticommuting Pauli elements. The following corollary formalizes this reduction.

\begin{Corollary}\label{ueJu=e1_ueKu=en+1}
    Let $b,c\in\mathcal{P}_n^+$ such that $b$ and $c$ anticommute, then 
    \begin{align}
        \exists u \in \mathcal{R}_n:ubu^\dagger =e_1 \quad and\quad ucu^\dagger=e_{n+1}
    \end{align}
\end{Corollary}
\begin{proof}
    By Lemma~\ref{rotor_e1}, we may first assume that $b=e_1$. The proof then consists of constructing Pauli rotors that leave $e_1$ invariant while progressively transforming $c$ into $e_{n+1}$. Depending on the parity of the support of $c$, an auxiliary vector may be introduced to reduce the problem to the even case. The complete construction is given in \mbox{Appendix \ref{ueJu=e1_ueKu=en+1_proof}}.
\end{proof}
\begin{Proposition}\label{prop1}
    Let $b_1,\dots,b_r \in \mathcal{P}_n$ and $\varepsilon_1,\dots,\varepsilon_r \in\{0,1\}$, consider $b = b_1^{\varepsilon_1}\cdots b_r^{\varepsilon_r}$, then
    \begin{equation}
        1_{\{b \in \mathcal{G}_n^+\}} = \sum_{j=1}^r \varepsilon_j1_{\{b_j \in \mathcal{G}_n^+\}}\mod 2
    \end{equation}
\end{Proposition}
\begin{proof}
    Immediate result of \eqref{odd_even_mv}.
\end{proof}

We are now ready to prove Theorem \ref{clifford_grp}.
The strategy is to use induction on the number of qubits.
The base case illustrates the structure of the Clifford group in the simplest nontrivial setting, while the induction step shows how Clifford operators can be decomposed into rotors and lower-dimensional Clifford transformations.
\begin{proof}[Proof of Theorem \ref{clifford_grp}]
    We shall use a proof by induction.

    \textbf{{Case} $n=1$:}
    Let $\lambda \in \mathcal{C}_1$, from Corollary \ref{ueJu=e1_ueKu=en+1}, we know that there exists $u \in \mathcal{R}_1$ such that $u(\lambda e_1\lambda^\dagger)u^\dagger =e_1$ and $u(\lambda e_2\lambda^\dagger)u^\dagger =e_2$. Let $\mu = u\lambda \in \mathcal{U}_1$, we have $\mu e_j\mu^\dagger = e_j$, so $\mu e_j = e_j\mu$ for $j=1,2$. Because $\mu$ commutes with $e_1$ and $e_2$, it commutes with every element in $\mathcal{G}_1$, however the only linear operator that commutes with every other linear operator is a scalar operator, since $\mu$ is also unitary, then $\mu =e^{i\theta}$, which means that $\lambda =u^\dagger \in \mathcal{R}_1$ up to a \mbox{global phase.}

    Having established the result for a single qubit, we now assume that the theorem holds for $(n-1)$ qubits and prove it for $n$-qubits.
    
    \textbf{{Case} $n>1$:}
    Let $\mathcal{G}_{n-1}$ be the Geometric Algebra, defined in Section \ref{complexification} associated with the $2(n-1)$ dimensional space $\mathbb{C}^{2(n-1)}=\text{Span}_\mathbb{C}\{e_2,\dots,e_n,e_{n+2},\dots,e_{2n}\}$. By the same notations we define $\mathcal{P}_{n-1}$, $\mathcal{C}_{n-1}$ and $\mathcal{R}_{n-1}$, the induction hypothesis is in that \mbox{Geometric Algebra.}

    Let $\lambda \in \mathcal{C}_n$, using corollary \ref{ueJu=e1_ueKu=en+1} we can safely assume that $\lambda e_1\lambda^\dagger =e_1$ and \linebreak  $\lambda e_{n+1} \lambda^\dagger = e_{n+1}$. This means that $\lambda $ commutes with both $e_1$ and $e_{n+1}$, we write 
    \begin{align}
        \lambda = \sum_{J \subset \{1,\dots,2n\}}\lambda_Je_J
    \end{align}
    {Since} $(e_J)_J$ is the canonical basis for $\mathcal{G}_n$, all the terms $e_J$ that anticommute with either $e_1$ or $e_{n+1}$ should not be in the sum, therefore $\lambda_J=0$, we know from Proposition \ref{prop_commutation} that the only $e_J$ that commute with both vectors are such that
    \begin{align}
        |J| = 1_{\{1 \in J\}} \mod2\quad  \text{ and }\quad |J| = 1_{\{n+1 \in J\}} \mod2
    \end{align}
    {This} means that we can write $\lambda$ as
    \begin{align}
        \lambda =& \sum_{J \subset \{2,\dots,n,n+2,\dots,2n\};|J|\text{ even}}\lambda_Je_J+\sum_{J \subset \{2,\dots,n,n+2,\dots,2n\};|J|\text{ odd}}\lambda_Je_{J\cup \{1,n+1\}}
        \end{align}
        \begin{align}\label{eq_lambda_proof_4}
        \lambda=\lambda^+ + e_1e_{n+1}\lambda^-
        \end{align}
    where $\lambda^+ \in \mathcal{G}_{n-1}^+$ and $\lambda^- \in \mathcal{G}_{n-1}^-$. A fact that we will use a lot is that the bivector $e_1e_{n+1}$ commutes with elements of $\mathcal{G}_{n-1}$ since $e_1$ and $e_{n+1}$ either both commute or anticommute with $e_K \in \mathcal{G}_{n-1}$, depending on $|K|$. From the point of view of $\mathcal{G}_{n-1}$, $e_1e_{n+1}$ acts as the imaginary number $i$: it commutes with every element in the algebra, it verifies \linebreak  $(e_1e_{n+1})^2=-1$ and consequently $(e_1e_{n+1})^\dagger = -e_1e_{n+1}$, therefore we shall denote it by $i_G$, the geometric complex number: $i_G =e_1e_{n+1}$. We also write 
    $$\lambda = \lambda ^+ + i_G \lambda^-$$

    Our first instinct is to say that $\mathcal{P}_{n-1}$ is stable under the conjugation action of $\lambda$, but it turns out that $\lambda$ verifies a stability relation a bit more general than that; we will verify that 
    \begin{equation}
        \forall b \in \mathcal{P}_{n-1} : \lambda b\lambda^\dagger =i_G^\varepsilon c \quad c \in \mathcal{P}_{n-1}, \varepsilon\in\{0,1\}
    \end{equation}
    {In other} words, $\lambda$ stabilizes not the $n-1$-qubit Pauli group, but the $n-1$-qubit Pauli group augmented with the geometric complex number $i_G$.

    A necessary and sufficient condition for that to be true is 
    \begin{align}\label{NS condition}
    \forall j\in \{2,..,n,n+2,...,2n\}: \lambda e_j \lambda^\dagger =i_G^\varepsilon c \quad c \in \mathcal{P}_{n-1}, \varepsilon\in\{0,1\}
        \end{align}
        {Let} $ j \in \{2,...,n,n+2,...,2n\}$. Since the vectors $e_j, j \neq 1,n+1$ anticommute with both $e_1$ and $e_{n+1}$, $\lambda e_j \lambda^\dagger$ anticommutes also with both $\lambda e_1 \lambda^\dagger = e_1$ and $\lambda e_{n+1}\lambda^\dagger =e_{n+1}$, therefore by the same argument made to justify Equation \eqref{eq_lambda_proof_4}, $\lambda e_j\lambda ^\dagger = i^ae_K$ where either both $1$ and $n+1$ are in $K$ or neither. In other words, we have exactly $e_K = \pm i_G^\varepsilon e_{K'}$ where $K' \subset \{2,...,n,n+2,...,2n\}$, thus obtaining the Condition \eqref{NS condition}.

    Recall that $\lambda = \lambda ^+ + i_G \lambda^-$. Let us consider $\tilde{\lambda} \in \mathcal{G}_{n-1}$ defined as 
    \begin{equation}
        \tilde{\lambda} = \lambda^+ + i \lambda^-
    \end{equation}
    {We} prove that $\tilde{\lambda} \in \mathcal{C}_{n-1}$.
    First we have 
    \begin{align}
    \lambda \lambda^\dagger = (\lambda^+ +i_G\lambda^-)(\lambda^{+\dagger} -i_G \lambda^{-\dagger}) &= (\lambda^+\lambda^{+\dagger} + \lambda^- \lambda^{-\dagger}) + i_G(\lambda^-\lambda^{+\dagger} - \lambda^{+}\lambda^{-\dagger})=1
    \end{align}
    {So} 
    \begin{align}
        \lambda^+\lambda^{+\dagger} + \lambda^- \lambda^{-\dagger}=1 \quad;\quad  \lambda^-\lambda^{+\dagger} - \lambda^{+}\lambda^{-\dagger}=0
    \end{align}
    {Then}
    \begin{align}
    \tilde{\lambda} \tilde{\lambda}^\dagger = (\lambda^+ +i\lambda^-)(\lambda^{+\dagger} -i \lambda^{-\dagger}) &= (\lambda^+\lambda^{+\dagger} + \lambda^- \lambda^{-\dagger}) + i(\lambda^-\lambda^{+\dagger} - \lambda^{+}\lambda^{-\dagger})=1
    \end{align}
    {Thus} $\tilde{\lambda} \in \mathcal{U}_{n-1}$.
    
    Moreover, we know from previously that for $j\neq1,n+1$ $\lambda e_j \lambda^\dagger = (i_G)^{\varepsilon}b$ where $b\in \mathcal{P}_{n-1}$, on the other hand we have 
    \begin{align}
        \lambda e_j &=\lambda^+e_j + i_G \lambda^-e_j\\
        (i_G)^\varepsilon b\lambda &= (i_G)^\varepsilon b \lambda^+ +(i_G)^{\varepsilon+1}b\lambda^-
    \end{align}
    {So} depending on the value of $\varepsilon \in \{0,1\}$ we get 
    \begin{align}
     \begin{cases}
         \lambda^+e_j =b \lambda^+\\
         \lambda^-e_j=b\lambda^-
     \end{cases}
     \quad \text{or} \quad \begin{cases}
         \lambda^+e_j =-b \lambda^-\\
         \lambda^-e_j=b\lambda^+
     \end{cases}
    \end{align}
    {Which} is equivalent, same as before, to 
    \begin{align}
        \tilde{\lambda} e_j \tilde{\lambda}^\dagger = i^{\varepsilon}b
    \end{align}
    {Thus} allowing us to conclude that $\tilde{\lambda}\in\mathcal{C}_{n-1}$. We use the induction hypothesis to write
    \begin{equation}
        \tilde\lambda = \prod_{k=1}^r \exp(\frac{\pi}{4}b_k) \quad b_k \in \mathcal{P}_{n-1} \subset \mathcal{P}_n
    \end{equation}
    {Expanding} this product gives
    \begin{align}
        \tilde\lambda= \sum_{\varepsilon_1,\dots,\varepsilon_r \in \{0,1\} }\frac{1}{\sqrt2^r}b_1^{\varepsilon_1}\cdots b_r^{\varepsilon_r}
    \end{align}
    {Using} Proposition \ref{prop1}, we get 
    \begin{align}
        \lambda^+ = \sum_{\substack{\varepsilon_1,\dots,\varepsilon_r \in \{0,1\} \\\sum_{j=1}^r \varepsilon_j1_{\{b_j \in \mathcal{G}_n^+\}}=0\mod 2}}\frac{1}{\sqrt2^r}b_1^{\varepsilon_1}\cdots b_r^{\varepsilon_r}\\
        \lambda^- = \sum_{\substack{\varepsilon_1,\dots,\varepsilon_r \in \{0,1\} \\\sum_{j=1}^r \varepsilon_j1_{\{b_j \in \mathcal{G}_n^+\}}=1\mod 2}}\frac{-i}{\sqrt2^r}b_1^{\varepsilon_1}\cdots b_r^{\varepsilon_r}\\\nonumber
    \end{align}
    {So} we rewrite $\lambda$ as \vspace{-3pt}
    \begin{align}\nonumber
        \lambda=& \sum_{\substack{\varepsilon_1,\dots,\varepsilon_r \in \{0,1\} \\\sum_{j=1}^r \varepsilon_j1_{\{b_j \in \mathcal{G}_n^+\}}=0\mod 2}}\frac{1}{\sqrt2^r}b_1^{\varepsilon_1}\cdots b_r^{\varepsilon_r} + \sum_{\substack{\varepsilon_1,\dots,\varepsilon_r \in \{0,1\} \\\sum_{j=1}^r \varepsilon_j1_{\{b_j \in \mathcal{G}_n^+\}}=1\mod 2}}i_G\frac{-i}{\sqrt2^r}b_1^{\varepsilon_1}\cdots b_r^{\varepsilon_r} \\\nonumber
        =&\sum_{\substack{\varepsilon_1,\dots,\varepsilon_r \in \{0,1\} \\\sum_{j=1}^r \varepsilon_j1_{\{b_j \in \mathcal{G}_n^+\}}=0\mod 2}}\frac{(-i\cdot i_G)^{\sum_{j=1}^r \varepsilon_j1_{\{b_j \in \mathcal{G}_n^+\}}}}{\sqrt2^r}b_1^{\varepsilon_1}\cdots b_r^{\varepsilon_r}\quad + \\&\qquad \sum_{\substack{\varepsilon_1,\dots,\varepsilon_r \in \{0,1\} \\\sum_{j=1}^r \varepsilon_j1_{\{b_j \in \mathcal{G}_n^+\}}=1\mod 2}}\frac{(-i\cdot i_G)^{\sum_{j=1}^r \varepsilon_j1_{\{b_j \in \mathcal{G}_n^+\}}}}{\sqrt2^r}b_1^{\varepsilon_1}\cdots b_r^{\varepsilon_r}\\\nonumber
        =& \sum_{\varepsilon_1,\dots,\varepsilon_r \in \{0,1\}} \frac{1}{\sqrt{2}^r}\prod_{k=1}^{r}\left((-i\cdot i_G)^{1_{\{b_k \in \mathcal{G}_n^+\}}}b_k\right)^{\varepsilon_k}\\\nonumber
        =&\prod_{k=1}^r\exp\left(\frac{\pi}{4}(-ie_1e_{n+1})^{1_{\{b_k \in \mathcal{G}_n^+\}}}b_k\right) \in \mathcal{R}_n \nonumber
    \end{align}
    {Which} completes the proof.
\end{proof}

Theorem \ref{clifford_grp} shows that every Clifford operator can be expressed as a product of Pauli rotors.
However, this construction does not provide information on the uniqueness or structure of such decompositions.
The following proposition refines this result by constraining the form of the Pauli elements appearing in such products.

\begin{Proposition}\label{prop_decomp}
    Let $\lambda \in \mathcal{C}_n$, then there exists $b_1,\dots,b_r \in \mathcal{P}_n^-$ with $b_j \neq \pm b_k$ and $b\in\mathcal{P}_n$ \linebreak  such that 
    \begin{equation}\label{equation112}
        \lambda = \left(\prod_{j=1}^r\rho_{b_j} \right)b
    \end{equation}
\end{Proposition}
\begin{proof}
    As always, we prove by induction the following statement
    \begin{align}
        \mathcal{H}_r :\text{" }&\forall b_1,\dots b_r \in\mathcal{P}_n^-,\exists s\leq r, c_1,\dots,c_s \in\mathcal{P}_n^-,c\in \mathcal{P}_n:\\
        &c_j\neq \pm c_k\\
        &\prod_{j=1}^r \rho_{b_j} = \left(\prod_{j=1}^s \rho_{c_j}\right)c \text{ "}
    \end{align}

    For $r=1$, we take $s=r$, $c_1 = b_1$ and $c=1$. 
    
    For $r \geq 1$ assume $\mathcal{H}_r$, let us prove $\mathcal{H}_{r+1}$, let $b_1,\dots,b_{r+1} \in \mathcal{P}_n^-$. From $\mathcal{H}_{r}$ we know that there exists $s\leq r, c_1,\dots,c_s \in\mathcal{P}_n^-,c\in \mathcal{P}_n |\quad c_j\neq \pm c_k$ such that
    \begin{equation}
        \prod_{j=2}^{r+1}\rho_{b_j} = \left(\prod_{j=1}^s \rho_{c_j}\right)c
    \end{equation}
    {If} $b_1 \neq \pm c_j$ for every $j$, then the proof is over; otherwise, there exists a unique $j_1$ such that $b_1 = \pm c_{j_1}$. We therefore have
    \begin{align}
        \lambda = \rho_{c'_{1}}\cdots\rho_{c'_{j_1-1}}(\rho_{b_1}\rho_{c_{j_1}})\rho_{c_{j_1+1}}\cdots \rho_{c_s}c
    \end{align}
    where $c'_k = \rho_{b_1}c_k\rho_{b_1}^{\dagger} = b_1^{\varepsilon_k}c_k$ (Since generally, we have $\rho_b\rho_a = \rho_{\rho_ba\rho_b^\dagger}\rho_b$). We know that $\rho_{b_1}\rho_{c_{j_1}}$ is a Pauli operator (either $1$ or $c_{j_1}$), and we can take it to the other side, with a possible change of sign inside the following rotors. For the other terms we get 
    \begin{align}
        \lambda = \rho_{c'_{1}}\cdots\rho_{c'_{j_1-1}}\rho_{\pm c_{j_1+1}}\cdots \rho_{\pm c_s}((\rho_{b_1}\rho_{c_{j_1}})c)
    \end{align}
    {The} product of rotors is of length $s-1 \leq r-1$, we therefore can use $\mathcal{H}_{r}$ to rewrite it, giving us the desired result.
\end{proof}

The above proposition gives us a sustainable way of constructing Clifford operators without needing to repeat the same rotor or its conjugate in the product. It also gives us an upper bound on the length of the decomposition $r$, since we have no more than $2^{2n}$ elements in $\mathcal{P}_n^-$ modulo $-1$. What it does not provide is information about the uniqueness of the decomposition; indeed, such uniqueness is not obtained in general. Take the \mbox{following example}
$$\lambda = \rho_{ie_1} \rho_{ie_2}$$
{Since}
$$ \rho_{ie_1} \rho_{ie_2}\rho_{-ie_1}= \rho_{-e_1e_2}$$
then $\lambda$ has an alternate decomposition
$$\lambda= \rho_{-e_1e_2} \rho_{ie_1}$$

\subsection{Clifford + T Universality in Geometric Algebra}
As previously mentioned in the beginning of this section, the Clifford group is not sufficient for approximating quantum gates, one needs, according to the \textit{Solovay-Kitaev Theorem}, to add a \textit{non-Clifford} gate to the collection \cite{N&C2010,Dawson and Nielsen (2005)}, typically one uses the $T$ gate, to make the set universal, this is why it is important to look for an equivalent of the $T_j$ gate in GA, the following proposition gives us exactly that.
\begin{Proposition}
    Let $T_j$ be the $T$ gate applied to the $j$-th qubit. Then 
    $$t_j=\lambda_{T_j} = \exp(\frac{\pi}{8}e_je_{n+j})$$
\end{Proposition}
\begin{proof}
    The $T$ gate applied on the $j$-th qubit is defined as $T_j = \exp(-i\frac{\pi}{8}Z_j)$, from Section \ref{pauli} we also know that the $Z$ gate is given by the multivector $ie_1e_2 = ff^\dagger - f^\dagger f$, so
    $$\lambda_{Z_j} = 1\otimes \cdots (\lambda_Z)_j\cdots\otimes1=1\otimes \cdots (f_jf_j^\dagger - f_j^\dagger f_j)\cdots\otimes1$$
    {Notice} that we do not have terms $f_k$ or $f_k^\dagger$ in the tensor product since $1 = f_kf_k^\dagger + f_k^\dagger f_k $; therefore, Theorem \ref{tensor} shows that the previous tensor product is a mere geometric product; in other words, we get
    $$\lambda_{Z_j} = 1 \cdots (f_jf_j^\dagger - f_j^\dagger f_j)\cdots1=ie_je_{n+j}$$
    {Thus} 
    $\lambda_{T_j} =\exp\left(\frac{\pi}{8}e_je_{n+j}\right)$.
\end{proof}
The introduction of a single $\pi/8$ Pauli rotor to the Clifford group means that we can now work with a smaller angle for all rotors. The resulting group has a lower degree of discreteness as it is now generated by the $\pi/8$ Pauli rotors, as stated by the \mbox{following Theorem}.
\begin{Theorem}
    The group generated by Clifford+T is generated by the Pauli rotors of angle $\pi/8$
    \begin{equation}
        \langle \mathcal{C}_n\cup \{ t_j \quad |1 \leq j \leq n\}\rangle = \left\langle \left\{ \exp\left(\frac{\pi}{8}b\right)| \quad b\in \mathcal{P}_n^- \right\} \right\rangle
    \end{equation}
\end{Theorem}
    \begin{proof}
        We notice that for every Pauli operator 
    $b$ of square $-1$, there exists a Clifford operator $\lambda$ such that $b = \lambda ie_je_{j+n}\lambda^\dagger$ for some $j$, therefore $\exp({\frac{\pi}{8}b}) = \lambda t_j \lambda^\dagger$, therefore $\langle \mathcal{C}_n\cup \{ t_j \quad |1 \leq j \leq n\}\rangle \supseteq \left\langle \left\{ \exp\left(\frac{\pi}{8}b\right)| \quad b\in \mathcal{P}_n^- \right\} \right\rangle$.

    On the other hand, a Clifford operator is generated by Pauli rotors of angle $\pi/4$, which can be seen as the square of Pauli rotors of angle $\pi/8$, and $t_j$ is itself a $\pi/8$ Pauli rotor; hence, the other inclusion holds, giving us the equality of the two groups.
    \end{proof}

The previous Theorem provides a natural geometric interpretation of the universality of the Clifford+\(T\) gate set within the Geometric Algebra framework. Indeed, Clifford operators are generated by Pauli rotors of the form
\[
\rho_b = \exp\left(\frac{\pi}{4}b\right),
\]
whose conjugation action preserves the discrete set of Pauli blades. Geometrically, Clifford transformations therefore act as exact symmetry transformations of the discrete Pauli geometry: under conjugation, blades are mapped to other blades.

By contrast, the introduction of \(\pi/8\)-rotors fundamentally changes this structure. Writing
\[
\rho_{b}^{(\pi/8)}=\exp\left(\frac{\pi}{8}b\right),
\]
one obtains, for noncommuting Pauli blades \(b \in\mathcal{P}_n^-\) and \(c\),
\[
\rho_{b}^{(\pi/8)}c\rho_{b}^{(\pi/8)\dagger}
=
\frac{c+bc}{\sqrt2}
\]
which is no longer a single Pauli blade, but rather a linear combination of blades. Geometrically, this means that conjugation by \(\pi/8\)-rotors does not simply transport one discrete geometric direction into another, but instead continuously deforms Pauli directions into new multivector directions lying outside the original discrete Pauli structure. While Clifford conjugation preserves the finite geometry generated by Pauli blades, \(\pi/8\)-rotors generate new effective rotation axes obtained from superpositions of geometric primitives. Repeated compositions of such transformations therefore produce increasingly rich geometric directions that are no longer confined to the original discrete blade geometry.

This phenomenon is consistent with the Lie-theoretic foundations underlying universal quantum gate approximation and the Solovay--Kitaev Theorem \cite{N&C2010,Dawson and Nielsen (2005)}, where repeated compositions of noncommuting generators allow dense exploration of the unitary group. In the present framework, this mechanism admits a direct geometric interpretation: universality emerges from the continuous deformation of discrete Pauli directions into increasingly rich multivector rotation axes.

\subsection{Discussion on the Geometric Algebra Representation}

As observed in the proof of Theorem \ref{isomorphism}, the Witt basis elements $f_j$ and $f_j^\dagger$ act on computational basis states in a manner analogous to fermionic annihilation and creation operators. In particular, relations \eqref{fdagger_operation} and \eqref{f_operation} show that $f_j^\dagger$ creates excitations while \(f_j\) removes them, up to phase factors arising from the ordering conventions of the basis. Consequently, the algebraic rules satisfied by the Witt basis are essentially identical to the canonical anticommutation relations of fermionic operators.

This observation naturally raises the question of whether the results obtained in the present work could equally be formulated using the standard language of creation and annihilation operators rather than Geometric Algebra.

The answer is essentially affirmative. Theorem 1 establishes an isomorphism between linear operators acting on the quantum register and multivectors of the complex Geometric Algebra. Consequently, the matrix formalism, the fermionic creation-annihilation formalism, and the present Geometric Algebra formulation are algebraically equivalent. No additional algebraic structure or computational power is introduced, and all results presented here could in principle be reformulated within the standard second-\mbox{quantized framework.}

The interest of the Geometric Algebra approach is therefore not algebraic but geometric. Within this framework, Pauli operators acquire a direct interpretation as oriented subspaces represented by blades. Commutation relations become incidence relations between these subspaces, while Clifford operators appear naturally as symmetry transformations acting on the resulting discrete geometry.

From this perspective, the value of the Geometric Algebra formulation lies in its ability to interpret familiar quantum structures through geometric notions such as directions, planes, orientations, and rotations. The geometric interpretation of Clifford+T universality developed in the previous subsection provides an example of this viewpoint: universality can be understood as a progressive densification of accessible rotation directions rather than solely as an algebraic property of matrix generators.

Consequently, the present work should be viewed primarily as a geometric reformulation of standard quantum operator theory rather than as an alternative \mbox{computational model.}

\section{Rotor Decomposition Algorithm} \label{DecAlgo}

Theorem \ref{clifford_grp} proves the existence of Pauli rotor decompositions but does not provide a constructive procedure.

\begin{Definition}
    Let $A \in \mathcal{G}_n$, we write $A$ in the canonical basis
    $$A = \sum_{J \subset \{1,\dots,2n\}}a_J e_J \quad a_J \in \mathbb{C}$$
    {The} support of $A$ is defined as the set 
    $$\text{Supp}(A) = \{ e_J | a_J \neq 0 \}$$
    {We} also define the support size of $A$ as
    $$\sigma(A) = | \text{Supp}(A)|$$
\end{Definition}

\vspace{-6pt}

\begin{algorithm}[H]
\algorithmcaption{{Greedy} 
 Pauli Rotor Decomposition}\label{greedy_Pauli_decomp}
\begin{algorithmic}[1]

\Require A Clifford operator $\lambda \in \mathcal C_n$
\Ensure A decomposition
\[
\lambda =
\left(\prod_{j=1}^s \rho_{b_j}\right)b
\]
with $b_j \in \mathcal P_n^{-}$ and $b \in \mathcal P_n$

\State Initialize $D \leftarrow \emptyset$

\While{$\sigma(\lambda) \neq 1$}

    \State Choose
    \[
    b \in \mathcal P_n^{-} \backslash (D \cup-D)
    \]
    minimizing
    \[
    \sigma\!\left(
    \rho_{b}^\dagger\lambda 
    \right)
    \]

    \State Append $b$ to $D$

    \State Update
    \[
    \lambda
    \leftarrow
    \rho_{b}^\dagger\lambda
    \]

\EndWhile
\State Append $\lambda$ to $D$
\State Return $D = [ b_1,\dots,b_s,b]$

\end{algorithmic}
\end{algorithm}

In the implementation used for the experiments, a multivector is stored by its coefficients in the canonical blade basis, requiring \(O(4^n)\) memory. At each iteration, the greedy step scans the possible choices of \(b\in\mathcal P_n^-\). For each candidate, the product \(\rho_b^\dagger\lambda\) can be computed by applying the blade multiplication induced by \(b\), which permutes blade coefficients up to signs and phases, followed by support counting. This gives a naive cost of \(O(4^n)\) per candidate. Since the number of Pauli blades grows as \(O(4^n)\), one greedy iteration has naive cost \(O(16^n)\), and the full algorithm has cost \(O(s16^n)\), where \(s\) is the number of iterations. Thus, the present implementation is intended as an exploratory structural tool rather than an optimized Clifford synthesis procedure.

\subsection{Empirical Behavior of the Algorithm}
For the empirical evaluation of the decomposition algorithm, we generated random Clifford operators from random products of Pauli rotors. For each number of qubits ($n \in \{1,\dots,4\}$), and for each initial rotor length ($1 \leq \ell \leq L$), we sampled ($N$) random Clifford operators of the form

$$\lambda = \prod_{j=1}^{\ell} \rho_{b_j}$$
where each ($b_j \in \mathcal P_n^{-}$) was chosen uniformly at random.

Each generated Clifford operator was then processed by the greedy decomposition algorithm introduced in the previous section. For every sample, we recorded:

\begin{itemize}
\item whether the algorithm terminated successfully,
\item the average and maximum lengths of the decompositions produced, for each initial length $1 \leq l \leq L$,
\item the evolution of $\sigma(\lambda)$ with iterations of the algorithm
\end{itemize}

The experiments therefore consisted of a total of $4 \times L \times N$ independent decomposition tests. The obtained data was used to study the empirical convergence of the algorithm and the growth behavior of decomposition lengths as functions of both the initial rotor length and the number of qubits. We took $L=20$ and $N=100$.

In order to better understand the behavior of the greedy reduction procedure, we tracked the evolution of the Pauli support size throughout the decomposition process for several representative Clifford operators.

All tested decompositions successfully reconstructed the original Clifford operators. Figure \ref{fig:avg_decomp} summarizes the average behavior of the decomposition lengths as functions of both the initial rotor length and the number of qubits.

\begin{figure}[H]
\includegraphics[width=0.9\linewidth]{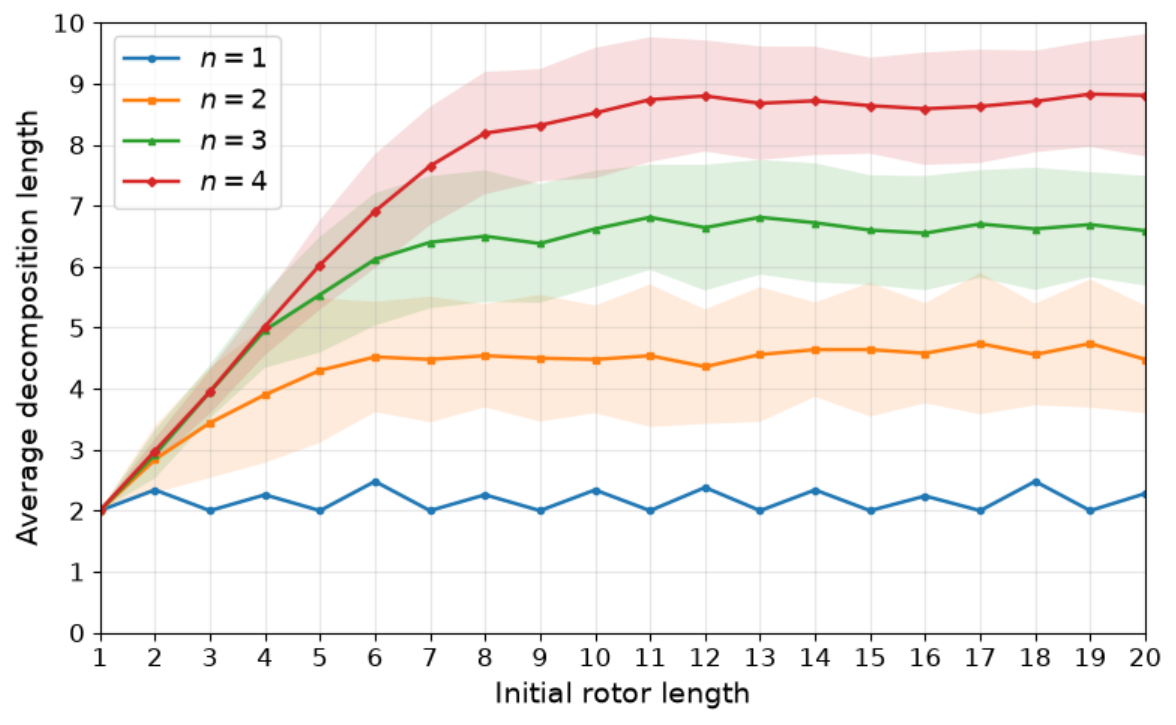}
\fcaption{Average
 decomposition length as a function of the initial rotor length.}
\label{fig:avg_decomp}
\end{figure}
Figure~\ref{fig:avg_decomp} shows that the average decomposition length grows much more slowly than the initial rotor length and appears to remain bounded by a quantity approximately linear in the number of qubits. Moreover, the largest decomposition length observed for each value of $n$ was equal to $2n+2$.

The shaded regions represent one empirical standard deviation around the mean and indicate that the decomposition lengths remain strongly concentrated around their average values. The case $n=1$ was omitted from this statistical representation since the very small size of $\mathcal{P}_1^{-}$ restricts the number of possible decomposition lengths, leading to noticeable finite-size oscillations that are not representative of the behavior observed for larger systems.

Figure~\ref{sigma_evolution} illustrates the evolution of the support size $\sigma(\lambda)$ as a function of the iteration number for two different randomly generated Clifford operators of different initial lengths. Although the support size is not strictly monotone during the reduction process, the overall behavior exhibits a rapid decay toward low-support configurations. In all tested examples, the algorithm terminated after a relatively small number of iterations compared to the initial support size.

\vspace{-3pt}
\begin{figure}[htbp]
    
    \includegraphics[width=0.9\linewidth]{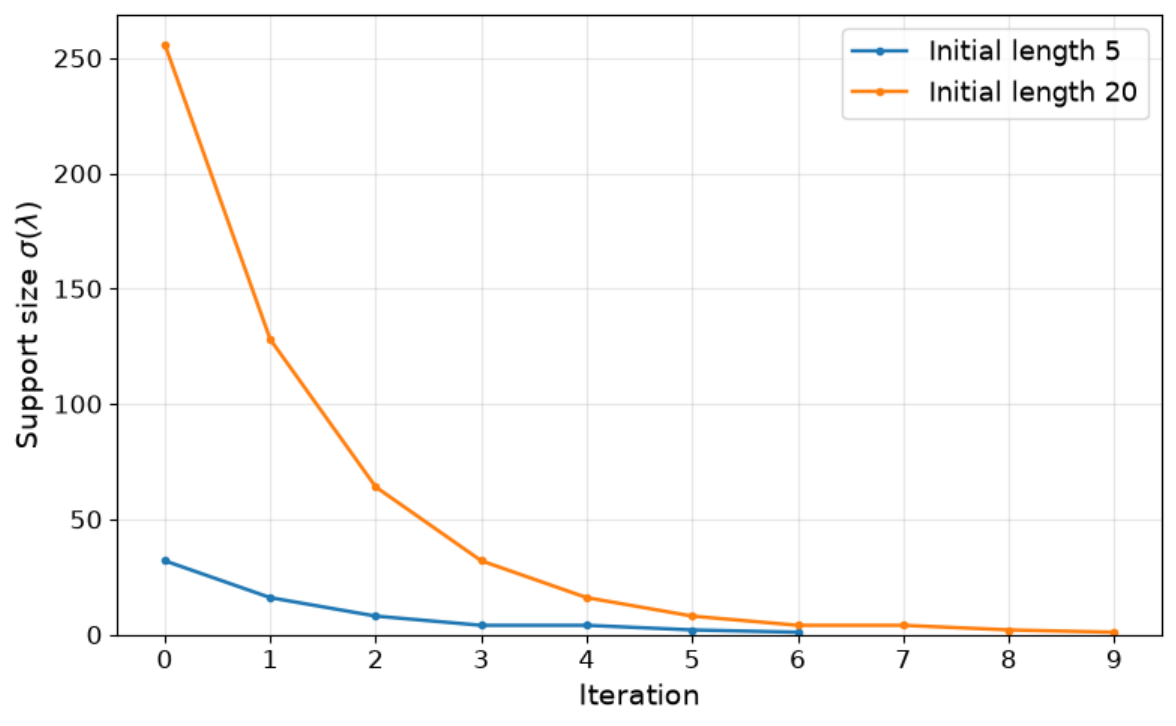}

    \fcaption{{Evolution} of $\sigma(\lambda)$ across iterations of the Algorithm~\ref{greedy_Pauli_decomp} for two random Clifford operators, $n=4$.}
    \label{sigma_evolution}

\end{figure}

\subsection{Structural Observations and Conjectures}
The empirical behavior of the decomposition algorithm suggests the existence of additional structural properties governing Clifford rotor decompositions.
\begin{conjecture}\label{conjecture_terminaison}
   For every Clifford operator $\lambda \in \mathcal C_n$, the greedy rotor decomposition algorithm (Algorithm~\ref{greedy_Pauli_decomp}) terminates after a finite number of steps. Moreover, the number of iterations required grows logarithmically with the Pauli support size. 
\end{conjecture}
One possible explanation for the observed behavior of the algorithm is the existence of algebraic redundancies among the Pauli blades appearing in a rotor decomposition. For example, if 
\[
\lambda=\rho_b\rho_c\rho_{bc},
\]
with $b \neq \pm c \in \mathcal{P}_n^-$ and $bc \in \mathcal{P}_n^{-}$, then $b$ and $c$ anticommute (otherwise \mbox{$(bc)^2 = b^2c^2 = 1$}) and one obtains
\[
\rho_b\rho_c=\rho_{bc}\rho_b
\]
so that
\[
\lambda=\rho_{bc}\rho_{bc}\rho_b=\rho_{-c}bc
\]

Thus, three rotors generated by dependent Pauli blades may be replaced by a decomposition involving only two rotors.

This observation suggests that minimal rotor decompositions may admit descriptions in terms of collections of Pauli blades free of such algebraic redundancies.

The following proposition shows that Conjecture~\ref{conjecture_terminaison} holds under this additional assumption.

\begin{Proposition}
Assume that a Clifford operator $\lambda$ admits a decomposition

\[
\lambda=\left(\prod_{j=1}^{r}\rho_{b_j}\right)b
\]
where $b\in\mathcal P_n$ and the blades
$b_1,\dots,b_r\in\mathcal P_n^-$ satisfy: \(
\nexists(\varepsilon_1,\dots,\varepsilon_r)
\in \{0,1\}^r\setminus\{0_r\}
\) such that \(
b_1^{\varepsilon_1}\cdots b_r^{\varepsilon_r}=\pm1
\). Then Algorithm~\ref{greedy_Pauli_decomp} terminates after at most $r$ iterations. Moreover, the number of iterations grows logarithmically with the support size.
\end{Proposition}

\begin{proof}
Expanding the product gives us
\[
\lambda=
\frac1{\sqrt2^r}
\sum_{\varepsilon\in \{0,1\}^r}
b_1^{\varepsilon_1}\cdots b_r^{\varepsilon_r}b
\]

By assumption, distinct vectors
$\varepsilon\in\{0,1\}^r$
produce distinct Pauli blades. Therefore $\sigma(\lambda)=2^r$. Let $c=b_1$, since $\rho_{c}^\dagger \rho_{b_1}=1$, one obtains

\[
\rho_c^\dagger\lambda
=
\left(\prod_{j=2}^{r}\rho_{b_j}\right)b
\]

Hence

\[
\sigma(\rho_c^\dagger\lambda)
=
2^{r-1}
<
2^r
=\sigma(\lambda)
\]

Repeating this argument removes one rotor at each iteration, so Algorithm~\ref{greedy_Pauli_decomp} terminates after at most $r$ steps, with $r=\log_2(\sigma(\lambda))$.
\end{proof}

Therefore, proving the existence of such irredundant rotor decompositions for arbitrary Clifford operators would imply Conjecture~\ref{conjecture_terminaison}.

Interestingly, the additional independence assumption introduced above is also consistent with the second empirical observation concerning decomposition lengths. Indeed, if a Clifford operator admits a decomposition $\lambda=\left(\prod_{j=1}^r \rho_{b_j}\right)b$ with multiplicatively independent blades $b_1,\dots,b_r\in\mathcal{P}_n^-$, then the cardinality of such a family is \mbox{necessarily bounded.}

Since every Pauli blade of $\mathcal P_n$ can be written as a product of the $2n+1$ elementary generators, no multiplicatively independent family of Pauli blades can contain more than $2n+1$ elements. Consequently, one necessarily has $r \leq 2n+1$.

Thus, the same structural assumption that provides a partial justification for Conjecture~\ref{conjecture_terminaison} also naturally explains the linear growth observed experimentally and motivates the following conjecture.

\begin{conjecture}
    Let $\lambda \in \mathcal{C}_n$, then $\exists r \leq 2n+1, b_1,\dots,b_r \in \mathcal{P}_n^-|b_j \neq \pm b_k, b\in\mathcal{P}_n$ such that 
    $$\lambda =\left( \prod_{j=1}^{r}\rho_{b_j}\right)b$$
\end{conjecture}

Although the proposed decomposition procedure relies on a simple greedy strategy and is not intended as an optimized circuit synthesis algorithm, the numerical results suggest that Clifford operators admit compact descriptions in terms of Pauli rotors. The observed linear growth in decomposition length should therefore be interpreted as a structural property of the Geometric Algebra representation rather than as a statement about circuit complexity or implementation cost.

Indeed, unlike elementary gates such as \(H\), \(S\), and \(CNOT\), general Pauli rotors may themselves require additional compilation into hardware-native gate sets. Consequently, the decomposition lengths reported here are not directly comparable to standard gate-count metrics. Instead, they provide evidence that Clifford operators may be geometrically interpreted through a relatively small number of primitive rotations.
\section{Discussion and Conclusions}

In this work, we developed an intrinsic formulation of the Pauli and Clifford groups within the framework of complex Geometric Algebra inspired by the research proposed in \cite{Hrdina et al. (2022)}. The Pauli group was characterized as the group of canonical blades, allowing Pauli operators to be interpreted directly as geometric primitives such as oriented directions (vectors), planes (bivectors), and higher-dimensional objects. Within this framework, commutation relations acquire a natural geometric meaning: commuting Pauli operators correspond to geometrically compatible structures, in the sense that one structure is invariant to rotations in the other.

Building upon this interpretation, we showed that Clifford operators are products of $\pi/4$ Pauli rotors. Rather than defining the Clifford group solely through its action on the Pauli group under conjugation, Theorem \ref{clifford_grp} provides a direct constructive description of Clifford transformations in terms of elementary geometric rotations generated by Pauli blades. From this perspective, Clifford dynamics remains internal to the geometry of the Pauli group itself, and Clifford operators may therefore be viewed as discrete geometric symmetry transformations acting on oriented subspaces. Furthermore, Proposition \ref{prop_decomp} shows that these decompositions may always be reduced to products of distinct Pauli rotors up to a final Pauli operator, providing a more structured representation of Clifford operators and yielding a first nontrivial upper bound of $2^{2n}$ on decomposition length.

We further investigated the structure of these decompositions through a greedy Pauli rotor decomposition algorithm (Algorithm \ref{greedy_Pauli_decomp}). Although the existence Theorem alone only yields a combinatorial upper bound on decomposition length, the empirical results suggest that Clifford operators admit significantly more compact rotor representations than expected. In particular, the experimentally observed decomposition lengths appear to be bounded above by $2n+2$. These observations suggest the presence of additional hidden geometric structure governing Clifford operators and their rotor representations.

Finally, the Geometric Algebra framework also provides a geometric interpretation of the universality of the Clifford+\(T\) gate set. While Clifford rotors preserve the discrete geometry generated by Pauli blades, the introduction of \(\pi/8\)-rotors continuously deforms Pauli directions into multivector directions lying outside the original discrete structure. In this way, universality emerges geometrically from the progressive densification of accessible rotation directions inside the unitary group.

Several questions remain open in this research. One of the most natural problems concerns the uniqueness of Pauli rotor decompositions. Although decompositions are not unique in general, the empirical behavior observed throughout this work suggests the possible existence of a deeper hidden structure. More precisely, it appears plausible that if decompositions are restricted to the minimal length, then any two minimal decompositions of the same Clifford operator should generate the same subgroup of the Pauli group. From this perspective, a Clifford operator would not simply admit several equivalent decompositions but could instead be intrinsically associated with a unique underlying Pauli subgroup encoding its geometric structure. Understanding the algebraic and geometric meaning of these subgroups remains an open problem.

More broadly, we hope that this work contributes to the development of Geometric Algebra (GA) methods in quantum computing. One of the central motivations behind this approach is that GA provides direct geometric interpretations for objects that are traditionally introduced through highly combinatorial matrix formalisms. As quantum systems grow in complexity, purely matrix-based descriptions often become increasingly difficult to interpret structurally \cite{Babbush et al. (2025)}. By contrast, GA naturally encodes directions, planes, rotations, and higher-dimensional incidence structures within a unified framework.

For example, within the GA framework developed here, the stabilizer structure used in quantum error correction could acquire a direct geometric interpretation. Commuting stabilizers correspond to mutually compatible geometric primitives whose induced rotational structures preserve one another, and this viewpoint naturally connects to syndrome decoding in quantum error correction. It is well known that optimal decoding rapidly becomes computationally intractable as the number of qubits increases \cite{Bradshaw&al2025,Iyer and Poulin (2013),Krastanov and Jiang (2017)}. Within the present framework, information about a syndrome could acquire an alternative interpretation in terms of geometric incompatibility relations between blades. Since commutation and anticommutation correspond to invariance properties between oriented subspaces, error syndromes may be viewed as signatures of geometric disturbances inside the Pauli structure. In this way, GA could offer additional intuition regarding the organization of syndrome spaces and the propagation of errors under Clifford transformations.

It is therefore conceivable that geometric viewpoints may provide alternative ways of understanding quantum circuits, error propagation, and operator structures beyond the standard computational perspective.

\nonumsection{Author Contributions}
Conceptualization, Y.A. and Z.T.; methodology, Y.A.; software, Y.A.; validation, Y.A. and Z.T.; formal analysis, Y.A.; investigation, Y.A. and Z.T.; resources, Y.A. and Z.T.; data curation, Y.A. and Z.T.; writing---original draft preparation, Y.A. and Z.T.; writing---review and editing, Y.A. and Z.T.; visualization, Y.A. and Z.T.; supervision, Y.A. and Z.T.; project administration, Z.T. All authors have read and agreed to the published version of the manuscript.

\nonumsection{Funding}
This research received no external funding.

\nonumsection{Data Availability Statement}
The code and data supporting the findings of this study are publicly available in the GitHub repository: \href{https://github.com/AmraouiY/Geometric-Algebra-Clifford-Gate-Rotor-Decomposition}{Geometric-Algebra-Clifford-Gate-Rotor-Decomposition}, accessed on 5 July 2026.

\nonumsection{Acknowledgments}
The authors are very grateful to the academic team of the student research track \textit{Parcours Recherche} at CentraleSupélec for enabling Y.A. to undertake this research project. The authors also thank Grégoire Veyrac for fruitful discussions on Geometric Algebra during the early stages of the project. Some of the ideas presented here were inspired by the previous project \textit{Pauli Rotors and Unitary Decomposition}, carried out at CentraleSupélec in 2020 by Antoine Cornillot, Joao Henrique Alves, and Jocelyn Terle, to whom the authors are sincerely grateful.

\nonumsection{Conflicts of Interest}
The authors declare no conflicts of interest.

\nonumsection{Abbreviations}
\begin{tabular}{@{}ll}
GA & Geometric Algebra\\
PDCS & Pauli Decomposition over Commuting Subsets
\end{tabular}

\section*{Appendix: Proofs}\label{app1}
\subsection{Proof of Proposition \ref{prop_commutation}} \label{prop_commutation_proof}
\begin{proof}
    For a vector $e_j$, we have 
    \begin{align}
        e_je_K = e_je_{k_1}\cdots e_{k_{|K|}}
    \end{align}
{Since} the vectors of the orthogonal basis anticommute, we obtain
\begin{align}
        e_je_K = (-1)^{|K|-1_{(j \in K)}}e_{k_1}\cdots e_{k_{|K|}}e_j = (-1)^{|K|-1_{(j \in K)}}e_Ke_j
    \end{align}
    {The} factor $(-1)^{|K|-1_{(j \in K)}}$ comes from the fact that $e_j$ will anticommute with every vector in $e_K$ except itself if $j\in K$. We therefore have
    \begin{align}
        e_Je_K = (-1)^ae_Ke_J
    \end{align}
    where
    \begin{align}
        a = \sum_{j\in J}|K|-1_{(j\in K)} = |J||K| -|K\cap J|
    \end{align}
    {So} we get the general formula for the commutation and anticommutation of elementary blades
    \begin{equation}
        e_Je_K = (-1)^{|K||J|-|K\cap J|}e_Ke_J
    \end{equation}
    which completes the proof.
\end{proof}

\subsection{Proof of Lemma \ref{rotor_b_effect}} \label{rotor_b_effect_proof}

\begin{proof}
Since $b^2=-1$, and using the series definition of the exponential, we get
\begin{align}\nonumber
    \rho_b &= \exp\left(\frac{\pi}{4}b\right)\\
    &=\cos\left(\frac{\pi}{4}\right) +\sin\left(\frac{\pi}{4}\right)b \\\nonumber
    &=\frac{1+b}{\sqrt2}
\end{align}
{So}
    \begin{align}\nonumber
        \rho_b\pi\rho_b^\dagger =& \left(\frac{1+b}{\sqrt{2}} \right)\pi\left( \frac{1-b}{\sqrt2}\right)\\
        =&\frac{\pi + b\pi -\pi b -b\pi b}{2}
    \end{align} 
    {Since} $b,\pi \in \mathcal{P}_n$, they either commute or anticommute, the first case yields $\frac{1}{2}(\pi + b\pi -b\pi -b^2\pi ) = \frac{1 }{2}(\pi +\pi) =\pi$, and the second case yields $\frac{1}{2}(\pi + b\pi +b\pi  +b^2\pi) = \frac{1}{2}(b\pi +b\pi) = b\pi$.
\end{proof}

\subsection{Proof of Lemma \ref{rotor_e1}} \label{rotor_e1_proof}

\begin{proof}
    Write $b=i^\varepsilon e_J$. 
    If $|J|$ is even (not empty, because otherwise $e_J=\pm1$), then we consider $b_1=i^{\varepsilon}e_{j_2} \cdots e_{j_{|J|}} =e_{j_1}b \in \mathcal{P}_n^-$, by Proposition \ref{prop_commutation} $b_1$ and $b$ anticommute, so \linebreak  $\rho_{b_1}b\rho_{b_1}^\dagger=b_1b=e_{j_1}b^2=e_{j_1}$. If $j_1\neq 1$, we then consider $b_2 = e_1e_{j_1}$ that anticommutes with $e_{j_1}$ and we thus have $\rho_{b_2}\rho_{b_1}b(\rho_{b_2}\rho_{b_1})^\dagger= b_2e_{j_1} = e_1$. We take $u = \rho_{b_2}\rho_{b_1}$, we have
    \begin{align}
        ubu^\dagger = e_1
    \end{align}
    
    If $|J|$ is odd, let $k \in \{1,\dots,2n\} \backslash J$, we consider $b_0 = ie_k$, by Proposition \ref{prop_commutation} $b_0$ and $b$ anticommute, so $\rho_{b_0}b\rho_{b_0}^\dagger = b_0b = \pm i^{\varepsilon+1}e_{J'}$ with $J' =J\cup \{k\}$, of even cardinal, we then go back to the first case.
\end{proof}
\subsection{Proof of Corollary \ref{ueJu=e1_ueKu=en+1}} \label{ueJu=e1_ueKu=en+1_proof}
\begin{proof}
    $b \neq \pm 1$ since otherwise it would commute with $c$. From the above Lemma, we can suppose that $b=e_1$ since the conjugation by a Clifford operator preserves product and commutation/anticommutation properties. Write $c = i^\varepsilon e_K$.

    The anticommutation of $e_1$ and $e_K$ is given by the following
    \begin{align}
        |K| = 1_{\{1 \notin K\}} \mod 2
    \end{align}
    
    If $|K|$ is even, then $1 \in K$.
    If $n+1 \in K$, then we take $c_1 = e_{n+1}c$, it anticommutes with $c$ since $c$ anticommutes with $e_{n+1}$ because $n+1 \in K$ and $|K|$ is even, and it commutes with $e_1$ since both $e_{n+1}$ and $c$ anticommute with it, so we deduce that $\rho_{c_1}c\rho_{c_1}^\dagger=c_1c =e_{n+1}$ and $\rho_{c_1}e_1\rho_{c_1}^\dagger=e_1$.
    Otherwise, take $c_1=ie_{n+1}e_1c$ and $c_2= ie_1$. Both commute with $e_1$, so both corresponding rotors stabilize $e_1$, moreover, $c_1$ anticommutes with $c$ since it commutes with $e_{n+1}$, so $\rho_{c_1} c\rho_{c_1}^\dagger =c_1c = ie_{n+1}e_1$, and that blade anticommutes with $c_2$, so $\rho_{c_2}\rho_{c_1}c(\rho_{c_1}\rho_{c_2})^\dagger = ie_1ie_{n+1}e_1 =e_{n+1}$.

    If $|K|$ is odd, then $1\notin K$, then we consider $c_0=ie_1$, we have $\rho_{c_0} c\rho_{c_0}^\dagger = ie_1i^\varepsilon e_K=i^{\varepsilon+1}e_{K \cup \{1\}}$, moreover, $\rho_{c_{0}}$ stabilizes $e_1$, we hence go back to the first case.

    We conclude that $\exists u \in \mathcal{R}_n$ such that $ue_1u^\dagger = e_1$ and $ucu^\dagger =e_{n+1}$; thus, the proof \linebreak  is complete.
\end{proof}


\nonumsection{References}


\end{document}